\documentclass{article}
\usepackage[utf8]{inputenc}
\usepackage{amsmath}
\usepackage{amsfonts}
\usepackage{amsthm}
\usepackage{amssymb}
\usepackage{mathtools}
\usepackage{dsfont}
\usepackage{geometry}
\usepackage{hyperref}
\usepackage{mathrsfs}
\usepackage{float}
\usepackage{subcaption}
\usepackage{xcolor}
\usepackage{siunitx}
\usepackage{booktabs}
\usepackage{authblk}

\newcommand{\NN}{\mathbb{N}}
\newcommand{\RR}{\mathbb{R}}
\newcommand{\PP}{\mathbb{P}}
\newcommand{\EE}{\mathbb{E}}
\newcommand{\Sp}{\mathbb{S}}
\newcommand{\blank}{{\mspace{2mu}\cdot\mspace{2mu}}}
\newcommand{\lp}{(}
\newcommand{\rp}{)}

\DeclareMathOperator{\conv}{conv}
\DeclareMathOperator*{\argmin}{arg\,min}

\theoremstyle{plain}
\newtheorem{theorem}{Theorem}
\newtheorem{lemma}{Lemma}
\newtheorem{proposition}{Proposition}
\newtheorem{corollary}{Corollary}
\theoremstyle{remark}
\newtheorem{remark}{Remark}
\theoremstyle{definition}
\newtheorem{definition}[theorem]{Definition} 
\newtheorem{example}{Example}

\title{Inverting Poisson-Laguerre tessellations}
\author{Thomas van der Jagt}
\author{Geurt Jongbloed}
\author{Martina Vittorietti}
\affil{Delft Institute of Applied Mathematics, Delft University of Technology.}
\date{\today}

\begin{document}

\maketitle

\begin{abstract}
    While it is well-known how to compute the cells of a Laguerre tessellation for a given set of weighted generator points, it is not obvious how to invert a Laguerre tessellation. That is, given that one observes a Laguerre tessellation, how can one retrieve the weighted generators corresponding to the observed cells. In this paper, we consider inversion of a class of random Laguerre tessellations known as Poisson-Laguerre tessellations. The weighted generators of observed cells of a Poisson-Laguerre tessellation are of interest because knowledge of these weighted generators is useful for statistical inference of Poisson-Laguerre tessellations. For general Laguerre tessellations we provide a characterization of all configurations of weighted generator points which yield the same Laguerre tessellation. For Poisson-Laguerre tessellations we propose a method for consistent inversion, meaning that as one observes the tessellation through increasing observation windows, a closer approximation of the original weighted generators can be obtained. In a simulation study we examine both performance of the inversion procedure, as well as the use of the obtained approximated weighted generators for nonparametrically estimating the weight distribution function corresponding to a Poisson-Laguerre tessellation. 
\end{abstract}

\section{Introduction}
Voronoi tessellations and their weighted generalizations have been successfully applied in a wide range of applications and fields. We refer to \cite{Okabe2000,Redenbach2025} for an overview of these tessellations and examples of their applications. In this paper we focus on a weighted generalization of the Voronoi tessellation known as the Laguerre tessellation. From an application point of view, Laguerre tessellations have shown to be accurate models of materials microstructures for a range of different materials. See for instance \cite{Falco2017}, \cite{Lautensack2008b}, \cite{Liebscher2015},  and \cite{Wu2010} for examples.

The Voronoi tessellation is defined via a set of points, referred to as the generators, and each generator corresponds to a cell in the tessellation. For Laguerre tessellations, the generators are weighted points. As opposed to generators in the Voronoi model, not all generators will necessarily generate a cell in the Laguerre setting. That is, the cell corresponding to some generator may be the empty set. For Voronoi- and Laguerre tessellations algorithms have been developed for computing the cells in the tessellation for a given set of generator points, see also \cite{Okabe2000} and \cite{Aurenhammer1987}. Additionally, in the case of the Voronoi tessellation it is known how to perform the inverse process. That is, given some description of the cells of a Voronoi tessellation it is possible to retrieve the generator points used to obtain this Voronoi tessellation. See for instance \cite{Schoenberg2003} for an algorithm to perform this procedure. In the case of Laguerre tessellations it is not possible to uniquely determine the set of weighted generator points that was used to obtain the Laguerre tessellation. More specifically, in \cite{Duan2014} it was shown via simulations that multiple configurations of weighted points may result in the same Laguerre tessellation, meaning that Laguerre tessellations are overparameterized. Questions related to the reconstruction and inversion of Laguerre tessellations have also been considered from several different perspectives in the literature. In \cite{bourne:hal-05008630,Bourne2025}, the authors study a related problem through semi-discrete convex order and establish reconstruction results for Laguerre tessellations based on cell volumes and centers of mass, revealing a close connection between reconstruction and inversion. In a more applied setting, \cite{Petrich2019} proposes a reconstruction procedure for polycrystalline materials from incomplete data, where weighted generator points are estimated using the Cross-Entropy Method. Furthermore, practical optimization-based approaches for fitting Laguerre-type polyhedral structures to experimental data have also been developed, see for instance \cite{Quey2018}.

In Figure \ref{overparameterization_example} two configurations of generators are shown which yield the same Laguerre tessellation. The circle around each generator point has radius equal to the weight of this generator. To the best of our knowledge there is no known result which presents a characterization of all configurations of weighted points which yield the same Laguerre tessellation. The process of retrieving generators from a given Voronoi- or Laguerre tessellation is referred to as inverting the Voronoi- or Laguerre tessellation, which explains the title of this paper. 

\begin{figure}[t!]
    \centering
    \makebox[\textwidth]{\makebox[\textwidth]{
    \begin{subfigure}[t]{0.5\textwidth}
        \centering
        \includegraphics[width=0.9\linewidth]{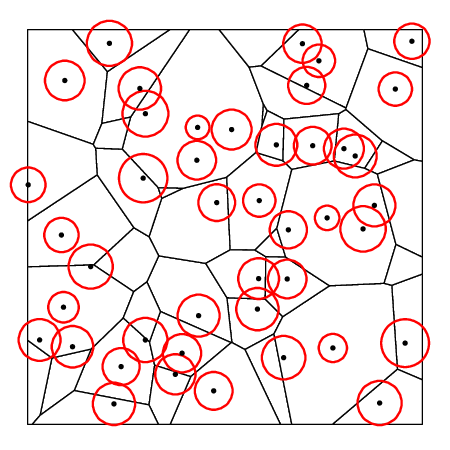}
    \end{subfigure}
    \begin{subfigure}[t]{0.5\textwidth}
        \centering
        \includegraphics[width=0.9\linewidth]{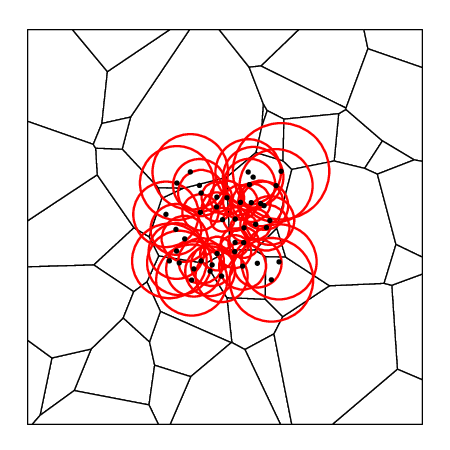}
    \end{subfigure}\hfill}}
    \caption{Example highlighting the overparameterization of a Laguerre tessellation.}\label{overparameterization_example}
\end{figure}

In the literature, various instances of random Laguerre tessellations have been studied, see \cite{Lautensack2007}, \cite{gusakova2025} and \cite{Seitl2021}. These random Laguerre tessellations are obtained by taking the weighted generator points to be a realization of a (marked) point process. The fact that it is not possible to uniquely recover the weighted generator points of a given Laguerre tessellation is somewhat disappointing in the context of statistical inference for random Laguerre tessellations. This being the case because the usual approach to performing statistical inference for random Laguerre tessellations is to consider the point process of weighted generators corresponding to observed cells to be known, and to then use these points to estimate the parameter(s) of the underlying point process model. This approach is for instance taken in \cite{Seitl2021} and \cite{vdjagt2025b}.

In this paper we first study general Laguerre tessellations, and then we shift our focus towards Poisson-Laguerre tessellations. For general Laguerre tessellations, we show that under a set of commonly satisfied regularity conditions it is possible to fully characterize all configurations of weighted generator points which yield the same Laguerre tessellation. Poisson-Laguerre tessellations were first studied in \cite{Lautensack2007} and \cite{Lautensack2008}. A Poisson-Laguerre tessellation in $\RR^d$ is obtained by taking the weighted generators to be a realization of a Poisson (point) process on $\RR^d \times E$ for some set $E \subset \RR$. Throughout, we consider the choice $E=(0,\infty)$, and we assume that the underlying Poisson process $\eta$ has intensity measure $\nu_d \times \mathbb{F}$. Here, $\nu_d$ is Lebesgue measure on $\RR^d$ and $\mathbb{F}$ is a locally finite measure on $(0,\infty)$. While a Poisson-Laguerre tessellation is a tessellation of $\RR^d$, in many practical settings it is common to take a so-called observation window $W \subset \RR^d$, and to only observe a part of the tessellation through this window. We provide sufficient conditions for being able to consistently invert the observed Poisson-Laguerre tessellation, as the observation window expands unboundedly to the whole space. This essentially means that while it is in general not possible to uniquely determine the weighted generator points of a Laguerre tessellation, if one observes a Poisson-Laguerre tessellation one can get a closer and closer approximation of the original weighted generator points corresponding to observed cells as one observes the Poisson-Laguerre tessellation through observation windows of increasing sizes. 

Taking these results into account, we are interested in whether the estimators for the distribution function $F(z) = \mathbb{F}((0,z])$, $z \geq 0$, as proposed in \cite{vdjagt2025b} still perform well if these estimators are computed based on our proposed approximation of the generator points instead of the true generator points. That is, we want to know if it is possible to estimate $F$, when we only rely on the cells of the tessellation which are (partially) observed through the observation window, without prior knowledge of the generators. We address this question via a simulation study.


In a practical context, it is often the case that one does not actually have Laguerre tessellation data. Instead, one may have some image data, for example a microscopic image of a materials microstructure. Then, one may fit a Laguerre tessellation to this image data. We refer to \cite{Alpers2025} and references therein for an overview of methods which can fit a unique Laguerre tessellation to image data. Typically, these procedures also provide a configuration of weighted generators which generate the fitted Laguerre tessellation. If one has a Laguerre tessellation without a corresponding configuration of weighted generators, then algorithm 1 in \cite{Duan2014} may be used to obtain such a configuration. Once a fitting procedure is performed one may apply the methodology proposed in this paper, to obtain a specific configuration of the weighted generators which is suitable for statistical analysis. This approach will especially be sensible if the image data at hand may approximately be considered a realization of a Poisson-Laguerre tessellation.


This paper is organized as follows. In section \ref{reconstruction_preliminaries} we introduce necessary notation and definitions. We present a characterization of the overparameterization of Laguerre tessellations in section \ref{section_overparameterization_laguerre}. Inspired by this result we propose a method for inverting a Poisson-Laguerre tessellation in section \ref{section_wls_reconstruction}. Sufficient conditions for consistent inversion of Poisson-Laguerre tessellations are given in section \ref{inversion_sufficient_conditions_section}. Then, we perform some simulations in section \ref{section_simulations_inversion}. Here, we apply the proposed inversion procedure and use the resulting approximation of the weighted generators to compute estimates of $F$. Finally, we provide some conclusions in section \ref{section_discussion_inversion}.


\section{Preliminaries}\label{reconstruction_preliminaries}
In this section we introduce necessary notation and definitions. Let $\nu_d$ denote Lebesgue measure on $\RR^d$, and $\sigma_{d-1}$ Lebesgue measure on the sphere $\Sp^{d-1} = \{x \in \RR^d: \Vert x \Vert = 1\}$, also known as the spherical measure. Given $x \in \RR^d$ and $r > 0$, we write $B(x, r) = \{y \in \RR^d : \Vert x-y\Vert < r\}$ and $\bar{B}(x, r) = \{y \in \RR^d : \Vert x-y\Vert \leq r\}$ for the open and closed ball respectively, with radius $r$ centered at $x$. Let $\kappa_d$ denote the volume of the $d$-dimensional unit ball:
\[\kappa_d := \nu_d\left(\Bar{B}(0,1)\right) = \frac{2\pi^\frac{d}{2}}{\Gamma\left(1 + \frac{d}{2} \right)}.\]
Let $A, B \subset \RR^d$, then the sum of sets is defined as: $A + B = \{a + b: a \in A, b \in B\}$. If $x \in \RR^d$, we also write: $A + x = \{a+x: a \in A\}$. For $c \in \RR$, we write $cA = \{ca:a\in A\}$.

\subsection{Point processes}
We now introduce several definitions related to point processes. For more background on the theory of point processes we refer to \cite{Daley2003}, \cite{Daley2008} and \cite{Last2018}. Suppose $\mathbb{X} = \RR^d \times (0,\infty)$, and let $(\Omega, \mathcal{A},\PP)$ be a probability space. A measure $\mu$ on $\mathbb{X}$ is locally finite if $\mu(B) < \infty$ for all bounded $B \in \mathcal{B}(\mathbb{X})$. Here, $\mathcal{B}(\mathbb{X})$ denotes the Borel $\sigma$-algebra of $\mathbb{X}$. Let $\mathbf{N}(\mathbb{X})$ denote the space of locally finite counting measures (integer-valued measures) on $\mathbb{X}$. We equip $\mathbf{N}(\mathbb{X})$ with the usual $\sigma$-algebra $\mathcal{N}(\mathbb{X})$, which is the smallest $\sigma$-algebra on $\mathbf{N}(\mathbb{X})$ such that the mappings $\mu \mapsto \mu(B)$ are measurable for all $B \in \mathcal{B}(\mathbb{X})$. A point process on $\mathbb{X}$ is a random element $\eta$ of $(\mathbf{N}(\mathbb{X}), \mathcal{N}(\mathbb{X}))$, that is a measurable mapping $\eta:\Omega \to \mathbf{N}(\mathbb{X})$. The intensity measure of a point process $\eta$ on $\mathbb{X}$ is the measure $\Lambda$ defined by $\Lambda(B):= \EE(\eta(B))$, $B \in \mathcal{B}(\mathbb{X})$. 

\begin{definition}
    Suppose $\Lambda$ is a $\sigma$-finite measure on $\mathbb{X}$. A Poisson process with intensity measure $\Lambda$ is a point process $\eta$ on $\mathbb{X}$ with the following two properties:
\begin{enumerate}
    \item For every $B \in \mathcal{B}(\mathbb{X})$, the random variable $\eta(B)$ is Poisson distributed with mean $\Lambda(B)$.
    \item For every $m \in \NN$ and pairwise disjoint sets $B_1,\dots,B_m \in \mathcal{B}(\mathbb{X})$, the random variables $\eta(B_1),\dots,\eta(B_m)$ are independent.
\end{enumerate}
\end{definition}

Let $\delta$ denote the Dirac measure, hence for $x \in \mathbb{X}$ and $B \in \mathcal{B}(\mathbb{X})$: $\delta_x(B) = \mathds{1}\{x \in B\}$. A counting measure $\mu$ on $\mathbb{X}$ is called simple if $\mu(\{x\}) \leq 1$ for all $x \in \mathbb{X}$. Similarly, a point process $\eta$ on $\mathbb{X}$ is called simple if $\PP\left(\eta(\{x\}) \leq 1, \text{ } \forall x \in \mathbb{X} \right) = 1$. Let $\mathbf{N}_s(\mathbb{X})$ be the subset of $\mathbf{N}(\mathbb{X})$ containing all simple measures. Define: $\mathcal{N}_s(\mathbb{X}):= \{A \cap \mathbf{N}_s(\mathbb{X}): A \in \mathcal{N}(\mathbb{X})\}$. Then, a simple point process on $\mathbb{X}$ is a random element $\eta$ of $(\mathbf{N}_s(\mathbb{X}), \mathcal{N}_s(\mathbb{X}))$. Simple point processes can be identified with their support, meaning that one may view the point process as a random set of discrete points in $\mathbb{X}$. We may, for example, write $x \in \eta$ instead of $x \in \mathrm{supp}(\eta)$. It is common practice to switch between the interpretations of a simple point process as a random counting measure or as a random set of points, depending on whichever interpretation is more convenient. We will also do this throughout this paper. Enumerating the points of a simple point process in a measurable way we may write:
\[\eta = \{x_1,x_2,\dots\}, \text{ \quad and \quad } \eta = \sum_{i =1}^{\eta(\mathbb{X})}\delta_{x_i}.\]
For $v \in \RR^d$ let $S_v$ denote the shift operator. Suppose $\eta = \{(x_1,h_1), (x_2, h_2),\dots\}$ is a point process with $x_i \in \RR^d$ and $h_i > 0$. Then, we define $S_v\eta := \{(x_1-v,h_1),(x_2-v,h_2),\dots\}$. We call a point process $\eta$ on $\RR^d$ stationary if $S_v\eta$ and $\eta$ are equal in distribution for all $v \in \RR^d$. Throughout, $(W_n)_{n \geq 1}$ denotes a so-called convex averaging sequence. That is, each $W_n \subset \RR^d$ is convex and compact, and the sequence is increasing: $W_n \subset W_{n+1}$. Finally, the sequence $(W_n)_{n \geq 1}$ expands unboundedly in all directions: $\sup\{r \geq 0: B(x,r) \subset W_n \text{ for some } x \in W_n\} \to \infty$ as $n \to \infty$.

\subsection{Laguerre tessellations} \label{prelim_laguerre}
In this section we describe general Laguerre tessellations as well as Poisson-Laguerre tessellations. By a tessellation of $\RR^d$ we mean a countable collection $T = \{C_i: i\in \NN\}$, of sets $C_i \subset \RR^d$ (the cells of the tessellation) satisfying some conditions. Specifically, the $C_i$'s have pairwise disjoint interiors, the union of all $C_i$'s equals $\RR^d$, the collection $T$ is locally finite and each $C_i$ is a compact and convex set with interior points. Local finiteness of $T$ means that only finitely many $C_i$'s have a non-empty intersection with any bounded $B \in \mathcal{B}(\RR^d)$. We should note that (Laguerre) tessellations are often defined on some bounded domain $A \subset \RR^d$, instead of $\RR^d$. Then, $T$ should be a finite collection instead and the union of the $C_i$'s equals $A$.

Now, we introduce the Laguerre diagram. Let $\varphi = \{(x_i,h_i)\}_{i \in \NN}$, with $x_i \in \RR^d$ and $h_i \in \RR$. Assume moreover that $x_i \neq x_j$ for $i\neq j$. The Laguerre cell associated with $(x,h) \in \varphi$ is defined as:
\begin{equation}
    C((x,h), \varphi) = \left\{y \in \RR^d: \Vert y-x \Vert^2 + h \leq \Vert y-x' \Vert^2 + h' \text{ for all } (x',h') \in \varphi\right\}. \label{eq_Laguerre_cell_def}
\end{equation}
For $i \in \NN$, the cell $C((x_i,h_i),\varphi)$ may be written as the intersection of half spaces:
\begin{equation}
    C((x_i,h_i),\varphi) = \bigcap_{j \in \NN}\left\{y\in\RR^d: 2\langle y,x_j - x_i\rangle \leq \Vert x_i \Vert^2 - \Vert x_j \Vert^2 +h_j - h_i\right\}. \label{laguerre_cel_halfspaces}
\end{equation}
Here, $\langle \blank,\blank\rangle$ is the usual inner product in $\RR^d$. The Laguerre diagram generated by $\varphi$ is the set of non-empty Laguerre cells, and is denoted by $L(\varphi)$: 
\[L(\varphi) := \left\{C((x,h), \varphi): (x,h) \in \varphi \text{ and } C((x,h), \varphi) \neq \emptyset\right\}.\]
A Laguerre diagram is not necessarily a tessellation, conditions on $\varphi$ are needed to ensure that $L(\varphi)$ is locally finite and that all cells are bounded. The random Laguerre diagrams we consider in this paper are in fact tessellations, as we will discuss in a moment. A Laguerre diagram has an interesting interpretation as a crystallization process. From the definition of a Laguerre cell it follows that:
\[x \in C((x_i,h_i),\varphi) \iff \exists t \geq h_i: x \in \bar{B}\left(x_i,\sqrt{t - h_i}\right) \text{ and } x \notin \bigcup_{j \neq i} B\left(x_j, \sqrt{(t-h_j)_{+}}\right),\]
with $(x)_{+} = \max\{x, 0\}$. Hence, we may associate the ball $B_i(t) := \bar{B}(x_i, \sqrt{(t-h_i)_{+}})$ with the pair $(x_i,h_i)$. This ball starts growing at time $t = h_i$. The ball initially grows fast, and then its growth slows down due to the square root. If $B_i$ is the first ball to hit a given point $x \in \RR^d$, then $x \in C((x_i,h_i),\varphi)$. It is possible that $x_i$ lies in another cell $C((x_j,h_j),\varphi)$, $i\neq j$ and yet $C((x_i,h_i),\varphi)$ may be non-empty. It can also be that $(x_i,h_i)$ does not generate a cell, essentially because its ball starts growing too late.

Let $P \subset \RR^d$ be a full-dimensional polytope. A subset $F$ of $P$ is called a face of $P$ if either $F=\emptyset$, $F=P$ or if there exists a supporting hyperplane $H$ of $P$ such that $F=P\cap H$. $H$ is a supporting hyperplane of $P$ if $H \cap P \neq \emptyset$ and if $P$ is contained in only one of the closed half spaces bounded by $H$. Note that each face of a convex polytope is again a convex polytope. A $k$-dimensional face of $P$ is called a $k$-face. Usually, $0$-faces are called vertices, $1$-faces edges and $(d-1)$-faces facets. A tessellation $T = \{C_i: i\in \NN\}$ is called face-to-face if the intersection of any two cells $C_i$ and $C_j$ is either empty or a $k$-face ($k \leq d-1$) of both cells. If moreover every $k$-face is contained in the boundary of exactly $d-k+1$ cells ($k=0,\dots,d-1$), then the tessellation is called normal. Let $\varphi = \{(x_i,h_i)\}_{i \in \NN} \subset \RR^d \times \RR$. We say that $\varphi$ satisfies the regularity conditions if:
\begin{enumerate}
    \item $\conv\{x_i: (x_i,h_i) \in \varphi\} = \RR^d$.
    \item Only finitely many $(x_i,h_i) \in \varphi$ satisfy $\Vert x_i - y \Vert^2 + h_i \leq t$ for any $y \in \RR^d$ and any $t \in \RR$.
\end{enumerate}
We say that the points of $\varphi$ are in general position if:
\begin{enumerate}
    \item A $(k-1)$-dimensional affine subspace contains at most $k$ $x_i$'s $(k = 2,\dots,d)$.
    \item At most $d+1$ points $(x,h) \in \varphi$ satisfy $\Vert x - y \Vert^2 + h = t$ for any $y \in \RR^d$, $t \in \RR$. 
\end{enumerate}
If $\varphi$ satisfies the regularity conditions, then the Laguerre diagram generated by $\varphi$ is a face-to-face tessellation. If the points of $\varphi$ are also in general position, then all cells of this Laguerre tessellation have dimension $d$ and this tessellation is normal. The aforementioned results can be found in \cite{Lautensack2007}. 

We are now ready to introduce Poisson-Laguerre tessellations. The Poisson-Laguerre tessellation is a generalization of the well-known Poisson-Voronoi tessellation, and was first studied in \cite{Lautensack2007} and \cite{Lautensack2008}. We will mostly follow the description of the Poisson-Laguerre tessellation as given in \cite{gusakova2025}. Throughout, $\eta$ is a Poisson process on $\RR^d \times(0,\infty)$ with intensity measure $\nu_d \times \mathbb{F}$. Here, $\mathbb{F}$ is a locally finite measure concentrated on $(0,\infty)$. In \cite{gusakova2025} it was shown that with probability one $\eta$ satisfies the regularity conditions and its points are in general position. As a consequence, $L(\eta)$ is with probability one a normal tessellation, which is also face-to-face. The random tessellation $L(\eta)$ is known as the Poisson-Laguerre tessellation. One may define the distribution function $F(z) = \mathbb{F}((0,z])$ for $z \geq 0$. Consistent estimators for $F$ based on weighted generators corresponding to observed cells were introduced in \cite{vdjagt2025b}. 

In this paper we will also use the following notation. As shown in Lemma 1 in \cite{vdjagt2025b}, for $x,y \in \RR^d$ and $h > 0$ we have $y \in C((x,h),\eta)-x \iff \eta(A_{x,h,y}) = 0$, with:
\begin{equation}
    A_{x,h,y} = \left\{(x',h') \in \RR^d \times (0,\infty): \Vert y \Vert^2 + h - h' > \Vert x + y-x' \Vert^2\right\}. \label{set_A_xhy}
\end{equation}
Moreover, 
\begin{equation}
    \PP\left(y \in C((x,h),\eta)-x\right) = \exp\left(-\kappa_d \int_0^{\Vert y \Vert^2 +h}\left(\Vert y \Vert^2 +h -t \right)^{\frac{d}{2}}\mathrm{d}F(t) \right). \label{point_inclusion_probability}
\end{equation}
For deterministic $(x,h) \in \RR^d \times \RR$, the set $C((x,h),\eta)$ is a so-called random closed set. Let $\mathcal{F}(\RR^d)$ denote the system of closed subsets of $\RR^d$. $\mathfrak{F}(\RR^d)$ denotes the $\sigma$-algebra on $\mathcal{F}(\RR^d)$ which is generated by all families $\mathcal{F}^K = \{F \in \mathcal{F}(\RR^d):F\cap K = \emptyset\}$, $K \in \mathcal{K}^d$. Here, $\mathcal{K}^d$ denotes the space of convex bodies in $\RR^d$, and a convex body is a convex and compact set with non-empty interior. Then, a random closed set is a random element of $(\mathcal{F}(\RR^d), \allowbreak\mathfrak{F}(\RR^d))$. Robbins' theorem states that for any random closed set $X \subset \RR^d$ and $p \in \NN$ we have:
\begin{equation}
    \EE\left(\nu_d(X)^p\right) = \int_{\RR^d}\dots\int_{\RR^d}\PP(y_1,\dots,y_p \in X)\mathrm{d}y_1 \dots \mathrm{d}y_p. \label{Robbins_thm}
\end{equation}
Robbin's theorem is essentially Fubini's theorem in the context of random closed sets. We refer to \cite{Molchanov2017} for more details on random closed sets and Robbins' theorem. From the proof of Theorem 6 in \cite{vdjagt2025b} we recall that for any $x \in \RR^d$:
\begin{align}
    &\int_0^\infty \EE\left(\nu_d(C((x,h),\eta)) \right)\mathrm{d}F(h) = \nonumber\\ 
    &= \int_0^\infty \int_{\RR^d}\PP\left(y \in C((x,h),\eta)-x\right)\mathrm{d}y\mathrm{d}F(h) \nonumber\\
    &=\int_0^\infty\int_{\RR^d}  \exp\left(-\kappa_d \int_0^{\Vert y \Vert^2 +h}\left(\Vert y \Vert^2 +h -t \right)^{\frac{d}{2}}\mathrm{d}F(t) \right)\mathrm{d}y\mathrm{d}F(h) = 1. \label{mean_cell_volume1}
\end{align}

\section{The overparameterization of Laguerre tessellations}\label{section_overparameterization_laguerre}

For a Laguerre tessellation it can be seen from the definition of a Laguerre cell that transforming all weights $h_i$ via $h_i \mapsto h_i + z$ for a fixed $z \in \RR$ does not change the tessellation. In \cite{Duan2014} it was shown via simulations that there also exist various other choices of generator points and weights which generate the same Laguerre tessellation. However, to the best of our knowledge there is no explicit characterization of the class of all configurations of points and weights which yield the same Laguerre tessellation. In this section, we obtain such a characterization for a class of Laguerre tessellations generated by weighted points satisfying commonly used regularity conditions. The following lemma from \cite{Meyron2019} highlights a large class of generator points and weights which result in the same Laguerre tessellation. 

\begin{lemma}[Proposition 6 in \cite{Meyron2019}]\label{overparameterization_laguerre_easy}
    Let $\mathcal{I} \subset \NN$ and $ \varphi =\{(x_i,h_i)\}_{i \in \mathcal{I}} \subset \RR^d \times \RR$. Let $\lambda > 0, c \in \RR^d$ and $z \in \RR$. For $i \in \mathcal{I}$ define:
\begin{align*}
    x_i' &:= \lambda x_i + c\\
    h_i' &:= \lambda h_i - \lambda(\lambda - 1)\Vert x_i \Vert^2 - 2\lambda \langle x_i, c\rangle + z.
\end{align*}
Set $\psi = \{(x_i',h_i')\}_{i \in \mathcal{I}}$. Then, $L(\varphi) = L(\psi)$. In fact, $C((x_i,h_i),\varphi) = C((x_i',h'_i),\psi)$ for all $i \in \mathcal{I}$.
\end{lemma}


\begin{remark}
    If we exclude the case $\lambda = 1$ in Lemma \ref{overparameterization_laguerre_easy}, then $x_i'$ and $h'_i$ may be written in the following form:
    \begin{align*}
    x_i' &:= \lambda (x_i - c') + c'\\
    h'_i &:= \lambda h_i - \lambda(\lambda - 1)\Vert x_i - c' \Vert^2 + z',
\end{align*}
for some $c' \in \RR^d$ and $z' \in \RR$. This form highlights in particular how each weight $h'_i$ depends on the distance of $x_i$ to $c'$.
\end{remark}

\begin{remark}
    Lemma \ref{overparameterization_laguerre_easy} was used to construct the example in Figure \ref{overparameterization_example}.
\end{remark}

In order to prove a converse of Lemma \ref{overparameterization_laguerre_easy} we need some additional assumptions. For instance, from the crystallic growth interpretation of a Laguerre tessellation as described in section \ref{prelim_laguerre}, we see that one can always add a weighted point to a configuration of weighted points with a sufficiently large weight such that this weighted point will generate an empty cell. After all, the cell corresponding to this weighted point will essentially "start growing too late". As such, we can at best characterize the class of weighted points which generate the same Laguerre tessellation, when restricting ourselves to the weighted points corresponding to the non-empty cells. In the theorem below we present such a characterization. 

\begin{theorem}\label{overparameterization_laguerre_hard}
    Let $\varphi = \{(x_i,h_i)\}_{i \in \NN} \subset \RR^d \times \RR$, $d \geq 2$. Assume that $\varphi$ satisfies the regularity conditions and the points of $\varphi$ are in general position. Let $\psi \subset \RR^d \times \RR$ be a countable set of distinct points. Assume that $C((x,h),\varphi) \neq \emptyset$ for each $(x,h) \in \varphi$ and $C((x,h),\psi) \neq \emptyset$ for each $(x,h) \in \psi$. Then, $L(\varphi) = L(\psi)$ if and only if $\psi = \{(x_i', h_i')\}_{i \in \NN}$ with:
    \begin{align*}
    x_i' &:= \lambda x_i + c\\
    h'_i &:= \lambda h_i + \lambda(\lambda - 1)\Vert x_i \Vert^2 + 2\lambda \langle x_i, c\rangle + z,
\end{align*}
for some $\lambda > 0$, $c \in \RR^d$ and $z \in \RR$.
\end{theorem}

\begin{proof}
    Let $i \in \NN$. If $x_i'$ and $h'_i$ are as in the statement of the theorem then $L(\varphi) = L(\psi)$ by Lemma \ref{overparameterization_laguerre_easy}. Now it remains to show the converse. Hence, we assume $L(\varphi) = L(\psi)$. Note that the points of $\psi$ can always be relabeled in such a way that we obtain $C((x_i,h_i),\varphi) = C((x_i',h'_i),\psi)$ for all $i \in \NN$. Throughout this proof we write $C_i = C((x_i,h_i),\varphi)$ and $C_i' = C((x_i',h'_i),\psi)$ for $i \in \NN$. Hence it remains to show that if $C_i = C_i'$ for all $i \in \NN$, then $x_i'$ and $h'_i$ are as in the statement of the theorem. Choose $i,j \in \NN$ with $i\neq j$ such that $C_i \cap C_j \neq \emptyset$. Since the Laguerre tessellation $L(\varphi)$ is a normal face-to-face tessellation, $C_i$ and $C_j$ share a facet. By (\ref{laguerre_cel_halfspaces}) this facet is contained within the supporting hyperplane:
    \[H_{ij} = \{x \in \RR^d: 2\langle x, x_i - x_j\rangle = \Vert x_i\Vert^2 - \Vert x_j\Vert^2 - h_j + h_i\}.\]
    We now seek conditions such that $\psi$ generates the same Laguerre cells. If so, we necessarily have $H_{ij} = H_{ij}'$ with:
    \[H_{ij}' = \{x \in \RR^d: 2\langle x, x_i' - x_j'\rangle = \Vert x_i'\Vert^2 - \Vert x_j'\Vert^2 - h'_j + h'_i\},\]
    this being the case because $H_{ij}'$ still needs to be a supporting hyperplane of $C_i'$ and $C_j'$. Note that the normal vectors of $H_{ij}$ and $H_{ij}'$ are given by $x_i - x_j$ and $x_i' - x_j'$ respectively. As a result there exists a $\lambda_{ij} \neq 0$ such that: $x_i' - x_j' = \lambda_{ij}(x_i - x_j)$. Now take another $k \in \NN$, with $k\neq i$ and $k \neq j$ such that $C_i$, $C_j$ and $C_k$ share a vertex. Due to normality of the Laguerre tessellation such a $k$ exists. As a result, $C_j$ and $C_k$ share a facet and $C_i$ and $C_k$ share a facet. Arguing as before, there exists $\lambda_{ik}\neq 0$ and $\lambda_{jk} \neq 0$ such that:
    \begin{align}
        x_i' - x_j' &= \lambda_{ij}(x_i - x_j) \label{eq_ij}\\
        x_i' - x_k' &= \lambda_{ik}(x_i - x_k) \label{eq_ik}\\
        x_j' - x_k' &= \lambda_{jk}(x_j - x_k). \label{eq_jk}
    \end{align}
    Suppose that $\lambda_{ik} \neq \lambda_{ij}$, we show that this leads to a contradiction. Equation (\ref{eq_ij}) may be written as: $\lambda_{ij}x_i = \lambda_{ij}x_j + x_i' - x_j'$. Subtracting (\ref{eq_jk}) from (\ref{eq_ik}) yields: $x_i' - x_j' = \lambda_{ik}(x_i - x_k) - \lambda_{jk}(x_j - x_k)$. Combining these two expressions yields:
    \[\lambda_{ij}x_i = \lambda_{ij}x_j + \lambda_{ik}(x_i - x_k) - \lambda_{jk}(x_j - x_k).\]
    By solving for $x_i$ we obtain:
    \[x_i = \left(1 - \frac{\lambda_{jk} - \lambda_{ik}}{\lambda_{ij} - \lambda_{ik}}\right)x_j + \left( \frac{\lambda_{jk} - \lambda_{ik}}{\lambda_{ij} - \lambda_{ik}}\right)x_k.\]
    This means that $x_i$ is a linear combination of $x_j$ and $x_k$, and therefore these three points are collinear. This is in contradiction with the assumption that the points of $\varphi$ are in general position. Hence, $\lambda_{ik} = \lambda_{ij}$. By symmetry: $\lambda \equiv \lambda_{ij} = \lambda_{ik} = \lambda_{jk}$. Because the Laguerre tessellation is a face-to-face tessellation we can iteratively consider neighboring cells such that we eventually find:
    \begin{equation}
        x_i' - x_j' = \lambda(x_i - x_j) \text{  for all } i,j \in \NN. \label{overparam_points_condition}
    \end{equation}
    Fix $i \in \NN$, and let $j \in \NN$. Choose $c \in \RR^d$ such that $x_i' = \lambda x_i + c$. From (\ref{overparam_points_condition}) we now obtain that also $x_j' = \lambda x_j + c$. Hence, $x_i' = \lambda x_i + c$ for all $i \in \NN$. We will argue that $\lambda > 0$ at the end of this proof. Now, we need to determine weights $h'_1, h'_2,\dots \in \RR$ such that the weighted points $\psi := \{(x_i',h'_i)\}_{i \in \NN}$. generate the same Laguerre cells. Consider once again $j\neq i$ such that $C_i \cap C_j \neq \emptyset$. Following the notation for $H_{ij}$ and $H_{ij}'$ as before, plugging in $x_i' = \lambda x_i + c$ we obtain:
    \[H_{ij}' = \{x \in \RR^d: 2 \langle x, \lambda(x_i - x_j)\rangle = \lambda^2 \Vert x_i \Vert^2 + 2\lambda \langle x_i, c\rangle - \lambda^2 \Vert x_j \Vert^2 - 2\lambda \langle x_j, c\rangle - h'_j + h'_i\}.\]
    It is now evident that $H_{ij} = H_{ij}'$ if and only if:
    \begin{align}
         & \lambda^2 \Vert x_i \Vert^2 + 2\lambda \langle x_i, c\rangle - \lambda^2 \Vert x_j \Vert^2 - 2\lambda \langle x_j, c\rangle - h'_j + h'_i = \lambda\left(\Vert x_i\Vert^2 - \Vert x_j\Vert^2 - h_j + h_i \right) \nonumber\\
        \iff& h'_i - h'_j = \lambda\left(h_i - (\lambda - 1)\Vert x_i \Vert^2 - 2\langle x_i,c\rangle \right) - \lambda\left(h_j - (\lambda - 1)\Vert x_j \Vert^2 - 2\langle x_j,c\rangle \right) \nonumber\\
        \iff& h'_i - h'_j = \lambda(f_i - f_j), \label{overparam_weights_condition}
    \end{align}
    with $f_i := h_i - (\lambda - 1)\Vert x_i \Vert^2 - 2\langle x_i,c\rangle$. Analogous to (\ref{overparam_points_condition}) we may argue that (\ref{overparam_weights_condition}) holds for all $i,j \in \NN$. Fix $i \in \NN$ and let $j \in \NN$, then choose $z \in \RR$ such that $h'_i = \lambda f_i + z$. From (\ref{overparam_weights_condition}) we now obtain that also $h'_j = \lambda f_j + z$. Hence, for all $i \in \NN$ we have:  
    \[h'_i = \lambda f_i + z = \lambda h_i - \lambda(\lambda - 1)\Vert x_i \Vert^2 - 2\lambda\langle x_i,c\rangle + z.\]
    This is precisely the form of $h_i'$ as in the statement of the theorem. Finally, we conclude that we must have $\lambda > 0$. This is the case because for $i,j \in \NN$ with $C_i \cap C_j \neq \emptyset$ the supporting hyperplane $H_{ij}$ also defines a half space containing $C_i$, recall (\ref{laguerre_cel_halfspaces}). Similarly, the half space induced by the supporting hyperplane $H_{ji}$ contains $C_j$. In order to have $C_i = C_i'$ for all $i \in \NN$ we need not only to preserve these supporting hyperplanes but also the orientation of the corresponding half spaces, which requires $\lambda > 0$. After all, choosing $\lambda < 0$ will flip the orientation of each half space.
\end{proof}

\begin{remark}
    The arguments in the proof of Theorem \ref{overparameterization_laguerre_hard} can also be used to show an analogous result for Laguerre tessellations in bounded domains. Then, $\varphi$ and $\psi$ should be finite sets of the same size. In the bounded case the points of $\varphi$ should also be in general position but $\varphi$ does not need to satisfy the regularity conditions.
\end{remark}

\section{Inverting Poisson-Laguerre tessellations via weighted least-squares}\label{section_wls_reconstruction}
In view of Theorem \ref{overparameterization_laguerre_hard} it is evident that whenever one observes a (random) Laguerre tessellation through a bounded window it is not possible to exactly reconstruct the original set of weighted generators corresponding to the observed cells. However, we can already intuitively understand that in many cases some configurations of the weighted points are more likely than others. Consider for example Figure \ref{overparameterization_example}. If one is given the information that the original generator points are uniformly distributed within the observation window, then the configuration in the left panel seems far more likely than the configuration in the right panel. This suggests that it may still be possible to get very close to the original configuration, even though exact reconstruction is not possible.

We now turn our attention towards Poisson-Laguerre tessellations. Recall that $\eta$ is a Poisson process on $\RR^d \times (0,\infty)$ with intensity measure $\nu_d \times \mathbb{F}$. Suppose we observe the Poisson-Laguerre tessellation $L(\eta)$ through a bounded observation window $W_n$, where $(W_n)_{n \geq 1}$ is a convex averaging sequence. In view of algorithm 1 in \cite{Duan2014} we can construct a configuration of weighted generators corresponding to the (partially) observed cells of the tessellation. Because we are in some sense interested in the "best" configuration of the weighted generators, and because all configurations are related to the original configuration via the form presented in Theorem \ref{overparameterization_laguerre_hard}, it does not matter which configuration we start with. As such we consider it to be of no loss of generality that we observe weighted generator points of $\eta^*$ up to a deterministic transformation of the form presented in Theorem \ref{overparameterization_laguerre_hard}. Here, $\eta^*$ denotes the set of the so-called extreme points of $\eta$, and is given by:
\[\eta^* = \{(x,h) \in \eta: C((x,h),\eta) \neq \emptyset\}.\]
Indeed, we need to stress that we can only hope to approximate the points of $\eta^*$ in an observation window, instead of the points of $\eta$. This being the case because there is no way to obtain any information on points of $\eta$ which do not generate a cell. Only points of $\eta^*$ can be determined up to a transformation of the form presented in Theorem \ref{overparameterization_laguerre_hard}.

Hence, we consider the following. For all points $(x,h) \in \eta^*$ with $x \in W_n$ we assume that the Laguerre cells corresponding to these points are observed. Additionally, the original points (elements of $\eta^*$) themselves are considered to be known up to a transformation of the form in Theorem \ref{overparameterization_laguerre_hard}. Specifically, suppose that $\lambda_0 > 0$, $c_0 \in \RR^d$, and define the function $f_0:\RR^d \times \RR \to \RR^d \times \RR$ via:
\[f_0(x,h) = \left(\frac{x}{\lambda_0} - \frac{c_0}{\lambda_0}, \frac{1}{\lambda_0} h + \frac{1}{\lambda_0}\left(\frac{1}{\lambda_0} - 1\right)\Vert x \Vert^2 - \frac{2}{\lambda_0}\left\langle x,\frac{c_0}{\lambda_0} \right\rangle \right).\]
We do not consider the additive constant $z$ in Theorem \ref{overparameterization_laguerre_hard}. We do this because the distribution of a Poisson-Laguerre tessellation is invariant under shifts of the distribution function $F(z) =\mathbb{F}((0,z])$. As such, we cannot do better than estimating $F$ up to a shift.

For any point $(x,h) \in \eta^*$ with $x \in W_n$ we observe $f_0(x,h)$ instead of $(x,h)$ where $\lambda_0$ and $c_0$ are unknown. Define: $W_n^0 = \smash{\frac{1}{\lambda_0}} W_n - \smash{\frac{c_0}{\lambda_0}}$. We wish to estimate $\lambda_0$ and $c_0$. For this purpose we define the following criterion function $T_n:\RR\times\RR^d \to [0,\infty)$:
{\allowdisplaybreaks
\begin{align}
    T_n(\lambda,c) :=& \frac{1}{\nu_d(W_n)}\sum_{(x,h) \in f_0(\eta)}\mathds{1}_{W_n^0}(x)\int_{C((x,h),f_0(\eta))} \Vert\lambda x + c - y \Vert^2 \mathrm{d}y \nonumber\\
    =& \frac{1}{\nu_d(W_n)}\sum_{(x,h) \in \eta}\mathds{1}_{W_n}(x)\int_{C((x,h),\eta)} \left\Vert\lambda \left(\frac{x}{\lambda_0} - \frac{c_0}{\lambda_0}\right) + c - y \right\Vert^2 \mathrm{d}y \label{def_criterion_function}
\end{align}
}%
Note that while $T_n$ contains a sum over elements of $\eta$ (or $f_0(\eta)$), only elements from $\eta^*$ (or $f_0(\eta^*)$) contribute to this sum. Intuitively, we believe that this criterion function is asymptotically minimized when taking $\lambda$ and $c$ such that $\lambda(x/\lambda_0 -c_0/\lambda_0) + c = x$ which corresponds to $\lambda=\lambda_0$ and $c=c_0$. That is, given the transformation $f_0$ we are interested in finding the inverse transformation. First, we compute the expected value of $T_n(\lambda,c)$. To this end, we need some additional notation. Let $y \in \RR^d$ and define the probability density function $p_F : \RR^d \to [0,\infty]$ via:
\[p_F(y) = \int_0^\infty \exp\left(-\kappa_d \int_0^{\Vert y \Vert^2 +h}\left(\Vert y \Vert^2 +h -t \right)^{\frac{d}{2}}\mathrm{d}F(t) \right)\mathrm{d}F(h).\]
The fact that $p_F$ is a probability density function follows from (\ref{mean_cell_volume1}). 
\begin{lemma}\label{pF_moments_lemma}
    For all $y \in \RR^d\setminus\{0\}$: $p_F(y) \leq 1/(\kappa_d \Vert y \Vert^d)$. Additionally, for all $q \in \NN$: $\int_{\RR^d}\Vert y \Vert^qp_F(y)\mathrm{d}y < \infty$. 
\end{lemma}
The proof of Lemma \ref{pF_moments_lemma} is given in section \ref{section_proofs_inversion}. Lemma \ref{pF_moments_lemma} ensures that the expected value of $T_n(\lambda,c)$, as given in the next lemma, is finite.
\begin{lemma}\label{criterion_function_expected_value_lemma}
    Let $\lambda > 0$ and $c \in \RR^d$, then:
    \[\EE\left(T_n(\lambda,c)\right) = \frac{1}{\nu_d(W_n)}\int_{W_n} \left\Vert \left(\frac{\lambda}{\lambda_0} - 1 \right)x + c - \frac{\lambda c_0}{\lambda_0} \right\Vert^2 \mathrm{d}x + \int_{\RR^d}\Vert y \Vert^2 p_F(y)\mathrm{d}y.\]
    And in particular:
    \[\EE\left(T_n(\lambda_0,c)\right) = \Vert c - c_0 \Vert^2 + \int_{\RR^d}\Vert y \Vert^2 p_F(y)\mathrm{d}y.\]
\end{lemma}
\begin{proof}
    Substituting $\tilde{y} = y-x$, and then writing $y$ instead of $\tilde{y}$ in (\ref{def_criterion_function}) we obtain:
    \begin{align*}
        T_n(\lambda,c) &=\frac{1}{\nu_d(W_n)}\sum_{(x,h) \in \eta}\mathds{1}_{W_n}(x)\int_{C((x,h),\eta)-x} \left\Vert \left(\frac{\lambda}{\lambda_0}-1\right)x -\frac{\lambda c_0}{\lambda_0} + c - y \right\Vert^2 \mathrm{d}y\\
        &= \frac{1}{\nu_d(W_n)}\sum_{(x,h) \in \eta}\mathds{1}_{W_n}(x)\int_{\RR^d} \left\Vert \left(\frac{\lambda}{\lambda_0}-1\right)x -\frac{\lambda c_0}{\lambda_0} + c - y \right\Vert^2 \mathds{1}\{y \in C((x,h),\eta)-x\} \mathrm{d}y
    \end{align*}
Note that $C((x,h),\eta) = C((x,h),\eta + \delta_{x,h})$. Hence, via the Mecke equation (Theorem 4.1. in \cite{Last2018}) and Fubini we obtain:
{\allowdisplaybreaks
\begin{align*}
    &\EE\left(T_n(\lambda,c) \right) = \\
    &=\frac{1}{\nu_d(W_n)} \int_{W_n}\int_{\RR^d} \left\Vert \left(\frac{\lambda}{\lambda_0}-1\right)x -\frac{\lambda c_0}{\lambda_0} + c - y \right\Vert^2 \int_0^\infty \PP\left(y \in C((x,h),\eta)-x \right)\mathrm{d}F(h)\mathrm{d}y\mathrm{d}x\\
    &= \frac{1}{\nu_d(W_n)} \int_{W_n}\int_{\RR^d} \left\Vert \left(\frac{\lambda}{\lambda_0}-1\right)x -\frac{\lambda c_0}{\lambda_0} + c - y \right\Vert^2 p_F(y)\mathrm{d}y \mathrm{d}x\\
    &= \frac{1}{\nu_d(W_n)} \int_{W_n} \left\Vert \left(\frac{\lambda}{\lambda_0}-1\right)x -\frac{\lambda c_0}{\lambda_0} + c \right\Vert^2 \left(\int_{\RR^d}p_F(y)\mathrm{d}y\right)\mathrm{d}x +\\
    &\phantom{=} \ -\frac{2}{\nu_d(W_n)}\left\langle \int_{W_n}\left(\frac{\lambda}{\lambda_0} -1\right)x -\frac{\lambda c_0}{\lambda_0} + c\mathrm{d}x, \int_{\RR^d}y p_F(y)\mathrm{d}y \right\rangle + \\
    &\phantom{=} \ + \frac{1}{\nu_d(W_n)} \int_{W_n}\mathrm{d}x\int_{\RR^d}\Vert y \Vert^2 p_F(y)\mathrm{d}y.
\end{align*}
}%
Because the probability density function $p_F$ is symmetric around $0$ and because\\ $\int_{\RR^d} \Vert y \Vert p_F(y)\mathrm{d}y < \infty$ by Lemma \ref{pF_moments_lemma}, it follows that $\int_{\RR^d} y p_F(y)\mathrm{d}y = 0$. As a consequence:
\[\EE\left(T_n(\lambda,c)\right) = \frac{1}{\nu_d(W_n)}\int_{W_n} \left\Vert \left(\frac{\lambda}{\lambda_0} - 1 \right)x + c - \frac{\lambda c_0}{\lambda_0} \right\Vert^2 \mathrm{d}x + \int_{\RR^d}\Vert y \Vert^2 p_F(y)\mathrm{d}y.\]
\end{proof}

\begin{remark}
    By considering any non-empty cell $C((x,h),\eta)$ with $(x,h)\in \eta$ and $x \in W_n$ to be fully observed we do not incorporate edge effects. We discuss this issue in section \ref{section_simulations_inversion}.
\end{remark}

\subsection{Definition and computation of an estimator}\label{section_definition_inversion}
From its expression, it is evident that the function $(\lambda, c) \mapsto \EE\left(T_n(\lambda,c)\right)$ attains its global minimum in $(\lambda_0,c_0)$. This inspires the definition of the following  estimators for $\lambda_0$ and $c_0$:
\[(\hat{\lambda}_n, \hat{c}_n) = \argmin_{(\lambda, c)\in \RR\times\RR^d} T_n(\lambda, c).\]
Applying this inversion procedure in practice means the following. For $\lambda \in \RR, c\in \RR^d$ define $f(\blank;\lambda,c):\RR^{d+1} \to \RR^d\times \RR$ via:
\[f((x,h);\lambda,c) = (\lambda x+c,\lambda h - \lambda(\lambda - 1)\Vert x \Vert^2 - 2\lambda \langle x, c\rangle).\]
Then, the inversion procedure boils down to computing the following point process:
\begin{equation}
    \hat{\eta}_n^* := f(f_0(\eta);\hat{\lambda}_n,\hat{c}_n) \cap (W_n \times(0,\infty)),\label{inversion_point_process}
\end{equation}
which may be considered an approximation of $\eta^* \cap (W_n \times(0,\infty))$. Let us provide some further motivation for estimating $(\lambda_0,c_0)$ by minimizing $T_n$. We intuitively expect that for a point $(x,h) \in \eta^*$, $x$ should be quite close to the center of its cell $C((x,h),\eta)$. Suppose that we consider for the center of a cell its centroid, also known as center of mass. We will now show that minimizing $T_n$ corresponds to minimizing the sum of the volume weighted squared distances of each generator to the center of mass of its cell. For a Borel set $K \subset \RR^d$ of positive volume its centroid $c(K)$ is defined as:
\[c(K) = \frac{1}{\nu_d(K)}\int_K x\mathrm{d}x.\]
The following equation highlights a convenient property of $c(K)$, for $a \in \RR^d$ we have:
\begin{equation}
    \int_K \Vert x - a \Vert^2 \mathrm{d}x = \int_K \Vert x - c(K) \Vert^2 \mathrm{d}x + \nu_d(K)\Vert a - c(K) \Vert^2. \label{centroid_property}
\end{equation}
Note that the LHS of (\ref{centroid_property}) is minimized by $a = c(K)$. In words, the centroid minimizes the integrated squared distance to a given set. Let us introduce the following notation, for the volume and the centroid corresponding to the Laguerre cell associated with $(x,h)$:
\[v_{x,h} := \nu_d(C((x,h),\eta)), \text{\quad} c_{x,h} := c(C((x,h),\eta))\]
Suppose that $(x',h') = f_0(x,h)$ for $(x,h) \in \eta^*$, because the function $f_0$ does not affect the resulting Laguerre tessellation we have:
\[v_{x',h'} := \nu_d(C((x',h'),f_0(\eta))) = \nu_d(C((x,h),\eta)) =: v_{x,h}.\]
Similarly we also have $\smash{c_{x',h'}} = \smash{c_{x,h}}$. We will frequently use this when we switch from summing over elements from $\eta$ to summing over elements from $\smash{f_0(\eta)}$, or vice-versa. Throughout, we will set $\smash{c_{x,h}} = 0$ whenever $\smash{v_{x,h}}=0$. Applying (\ref{centroid_property}) to the definition of $T_n$, we obtain:
\begin{align*}
    T_n(\lambda,c) &=\frac{1}{\nu_d(W_n)}\sum_{(x,h) \in f_0(\eta)}\mathds{1}_{W_n^0}(x)\left(\int_{C((x,h),f_0(\eta))} \Vert y - c_{x,h} \Vert^2 \mathrm{d}y + v_{x,h}\left\Vert \lambda x + c - c_{x,h} \right\Vert^2\right) \\
    &= \frac{1}{\nu_d(W_n)}\sum_{(x,h) \in \eta}\mathds{1}_{W_n}(x)\left(\int_{C((x,h),\eta)} \Vert y - c_{x,h} \Vert^2 \mathrm{d}y + v_{x,h}\left\Vert \lambda \left(\frac{x}{\lambda_0} - \frac{c_0}{\lambda_0}\right) + c - c_{x,h} \right\Vert^2\right).
\end{align*}
As such, we have written $T_n$ as the sum of two terms, and only the second term depends on $\lambda$ and $c$. Therefore, we also have:
\begin{equation}
    (\hat{\lambda}_n, \hat{c}_n) = \argmin_{(\lambda, c)\in \RR\times\RR^d}\frac{1}{\nu_d(W_n)}\sum_{(x,h) \in f_0(\eta)}\mathds{1}_{W_n^0}(x)v_{x,h}\left\Vert \lambda x + c - c_{x,h} \right\Vert^2.\label{estimators_wls_def}
\end{equation}
So indeed, $\smash{(\hat{\lambda}_n, \hat{c}_n)}$ minimizes the sum of volume weighted squared distances of the generators to the centers of mass of their cells. From (\ref{estimators_wls_def}), we can see that $T_n$ is a strictly convex function. As a consequence, we can obtain $\smash{(\hat{\lambda}_n, \hat{c}_n)}$ by computing the critical point of $T_n$ (indeed, there is only one critical point). Taking the partial derivatives of $T_n$ w.r.t. $\lambda$ and $c$ yields:

\begin{align*}
\frac{\partial T_n(\lambda, c)}{\partial \lambda} &= \frac{2}{\nu_d(W_n)}\sum_{(x,h) \in f_0(\eta)} \mathds{1}_{W_n^0}(x)v_{x,h}\left(\lambda \Vert x \Vert^2 +\left\langle c - c_{x,h},x \right\rangle \right) \\
\frac{\partial T_n(\lambda, c)}{\partial c} &= \frac{2}{\nu_d(W_n)}\sum_{(x,h) \in f_0(\eta)} \mathds{1}_{W_n^0}(x)v_{x,h}\left(\lambda x + c - c_{x,h} \right).
\end{align*} 
Setting these partial derivatives equal to zero and solving for $\lambda$ and $c$ yields a unique solution, which is given by:
{\allowdisplaybreaks
\begin{align}
\begin{split}
    \hat{\lambda}_n &= \frac{\frac{\sum_{(x,h) \in f_0(\eta)} \mathds{1}_{W_n^0}(x)v_{x,h}\langle c_{x,h},x\rangle}{\sum_{(x,h) \in f_0(\eta)} \mathds{1}_{W_n^0}(x)v_{x,h}} - \left\langle \frac{\sum_{(x,h) \in f_0(\eta)} \mathds{1}_{W_n^0}(x)v_{x,h} c_{x,h}}{\sum_{(x,h) \in f_0(\eta)} \mathds{1}_{W_n^0}(x)v_{x,h}}, \frac{\sum_{(x,h) \in f_0(\eta)} \mathds{1}_{W_n^0}(x)v_{x,h}x}{\sum_{(x,h) \in f_0(\eta)} \mathds{1}_{W_n^0}(x)v_{x,h}} \right\rangle}{\frac{\sum_{(x,h) \in f_0(\eta)} \mathds{1}_{W_n^0}(x)v_{x,h}\Vert x \Vert^2}{\sum_{(x,h) \in f_0(\eta)} \mathds{1}_{W_n^0}(x)v_{x,h}} - \left\Vert \frac{\sum_{(x,h) \in f_0(\eta)} \mathds{1}_{W_n^0}(x)v_{x,h}x}{\sum_{(x,h) \in f_0(\eta)} \mathds{1}_{W_n^0}(x)v_{x,h}} \right\Vert^2} \\
    \hat{c}_n &= \frac{\sum_{(x,h) \in f_0(\eta)} \mathds{1}_{W_n^0}(x)v_{x,h} c_{x,h}}{\sum_{(x,h) \in f_0(\eta)} \mathds{1}_{W_n^0}(x)v_{x,h}} - \hat{\lambda}_n \frac{\sum_{(x,h) \in f_0(\eta)} \mathds{1}_{W_n^0}(x)v_{x,h} x}{\sum_{(x,h) \in f_0(\eta)} \mathds{1}_{W_n^0}(x)v_{x,h}}. \label{cn_lambdan_formula}
\end{split}
\end{align}
}%
Or, if we write $f_0(\eta) \cap (W_n^0 \times \RR) = \{(x_1,h_1),\dots,(x_m,h_m)\}$ and $c_i = c_{x_i,h_i}$, $v_i = v_{x_i,h_i}$, then $\hat{\lambda}_n$ may be written more compactly as:
\begin{align*}
    \hat{\lambda}_n &= \frac{\left(\frac{\sum_{i=1}^m v_i\langle c_i,x_i\rangle}{\sum_{i=1}^m v_i}\right) - \left\langle \frac{\sum_{i=1}^m v_i c_i}{\sum_{i=1}^m v_i}, \frac{\sum_{i=1}^m v_ix_i}{\sum_{i=1}^m v_i} \right\rangle}{\left(\frac{\sum_{i=1}^m v_i\Vert x_i \Vert^2}{\sum_{i=1}^m v_i}\right) - \left\Vert \frac{\sum_{i=1}^m v_ix_i}{\sum_{i=1}^m v_i} \right\Vert^2}\\
    &= \frac{\left(\sum_{i=1}^m v_i \right)\left(\sum_{i=1}^m v_i\langle c_i,x_i\rangle\right) - \left\langle \sum_{i=1}^m v_i c_i, \sum_{i=1}^m v_ix_i \right\rangle}{\left(\sum_{i=1}^m v_i \right)\left(\sum_{i=1}^m v_i\Vert x_i \Vert^2\right) - \left\Vert \sum_{i=1}^m v_ix_i \right\Vert^2}.
\end{align*}
Similarly, $\hat{c}_n$ may be written as:
\[\hat{c}_n = \frac{\sum_{i=1}^m v_ic_i}{\sum_{i=1}^m v_i} - \hat{\lambda}_n\frac{\sum_{i=1}^m v_ix_i}{\sum_{i=1}^m v_i}.\]
In order for $\hat{\lambda}_n$ to be well-defined we need to verify that the denominator in its definition is not equal to zero. Whenever $x_1,\dots,x_m$ are distinct points, we have by Jensen's inequality:
\[\left\Vert \frac{\sum_{i=1}^m v_ix_i}{\sum_{i=1}^m v_i} \right\Vert^2 < \frac{\sum_{i=1}^m v_i\Vert x_i \Vert^2}{\sum_{i=1}^m v_i}.\]
Because $x_1,\dots,x_m$ are distinct points with probability one, $\hat{\lambda}_n$ is almost surely well-defined. As is frequently the case with least-squares estimators, these can often be computed via an explicit formula such as (\ref{cn_lambdan_formula}) or by solving a system of linear equations. In this case, note that we may write:
\[\frac{\sum_{i=1}^m v_i\left\Vert \lambda x_i + c - c_i \right\Vert^2}{\sum_{i=1}^m v_i} = \beta^TX\beta - 2u^T\beta + K.\]
Here, $X$ is a $(d+1)\times (d+1)$ matrix which is given by the following block matrix:
\[X = \begin{pmatrix}
        I_d & \left( \frac{\sum_{i=1}^m v_ix_i}{\sum_{i=1}^m v_i} \right)\\
        \left( \frac{\sum_{i=1}^m v_ix_i}{\sum_{i=1}^m v_i} \right)^T & \frac{\sum_{i=1}^m v_i\Vert x_i \Vert^2}{\sum_{i=1}^m v_i}
    \end{pmatrix}.\]
Each vector is to be interpreted as a column-vector and $\beta^T$ denotes the transpose of $\beta$. The matrix $I_d$ is the unit matrix in $\RR^{d\times d}$. Finally, $u$ and $\beta$ are vectors in $\RR^{d+1}$ and $K$ is a constant, these are given by:
\[u=\begin{pmatrix}
        \frac{\sum_{i=1}^m v_ic_i}{\sum_{i=1}^m v_i} \\
        \frac{\sum_{i=1}^m v_i\langle c_i,x_i\rangle}{\sum_{i=1}^m v_i}
    \end{pmatrix}, \ \beta = \begin{pmatrix}
        c\\
        \lambda
    \end{pmatrix}, \
K= \frac{\sum_{i=1}^m v_i\Vert c_i \Vert^2}{\sum_{i=1}^m v_i}.\]
Then, $(\hat{\lambda}_n, \hat{c}_n)$ may be computed by solving the linear system $X\beta = u$ for $\beta$. In simulations this was observed to be more numerically stable compared to using (\ref{cn_lambdan_formula}).

\subsection{Limiting behavior of the criterion function}

As a first important step towards understanding how $(\hat{\lambda}_n, \hat{c}_n)$ may behave as $n \to \infty$, we study how the criterion function $T_n$ behaves as $n \to \infty$. Before we can determine the limiting behavior of $T_n$, we need two lemmas which allow us to control the sizes of the Laguerre cells that are summed over in the definition of $T_n$.

\begin{lemma}\label{ball_in_cell_lemma}
    Let $r >0$ and $\varphi = \{(x_i,h_i)\}_{i \in \NN} \subset \RR^d \times \RR$. Let $(x,h) \in \varphi$. Define the following set:
    \begin{align*}
        D_{x,h} &= \left\{(x',h') \in \RR^d \times \RR: h'\leq r^2 + h, \ \Vert x - x'\Vert < r + \sqrt{r^2 + h - h'}\right\}.
    \end{align*}
    Then, if $(\varphi \setminus \{(x,h)\}) \cap D_{x,h} = \emptyset$ we have $\bar{B}(x,r) \subset C((x,h),\varphi)$.
\end{lemma}

\begin{proof}
    Suppose $(\varphi \setminus \{(x,h)\}) \cap D_{x,h} = \emptyset$. Let $y \in \bar{B}(x,r)$, it remains to show that $\Vert x- y\Vert^2 + h \leq \Vert x' - y\Vert^2 + h'$ for all $(x',h') \in \varphi \setminus \{(x,h)\}$. Let $(x',h') \in \varphi \setminus \{(x,h)\}$. If $h' \leq r^2 + h$ then by assumption $\Vert x - x'\Vert \geq r + \sqrt{r^2 + h - h'}$. As a consequence of this fact and $\Vert x - y \Vert \leq r$ we obtain:
    \[\Vert x - y\Vert^2 +h -h' \leq r^2 + h - h' \leq \left(\Vert x-x' \Vert -r \right)^2 \leq \left(\Vert y-x' \Vert + \Vert x-y \Vert -r \right)^2 \leq \Vert x' - y\Vert^2.\]
    So indeed, $\Vert x- y\Vert^2 + h \leq \Vert x' - y\Vert^2 + h'$. It remains to consider the case $h' > r^2 + h$. In that case we also obtain the desired result:
    \[\Vert x- y\Vert^2 + h \leq r^2 + h < h' \leq \Vert x' - y\Vert^2 + h'.\]
\end{proof}

The lemma below is essentially an intermediate result obtained in the proof of Proposition 3.1 in \cite{Flimmel2020}. For the sake of completeness, the proof is presented in section \ref{section_proofs_inversion}.

\begin{lemma}\label{cell_in_ball_lemma}
    Let $z > 0$, $R > 2\sqrt{z}$ and $\varphi = \{(x_i,h_i)\}_{i \in \NN} \subset \RR^d \times (0,z]$. Let $(x,h) \in \varphi$. Suppose that $J \in \NN$, and $\mathcal{C}_1,\dots,\mathcal{C}_J$ are convex cones with non-empty interiors with $\smash{\cup_{j=1}^J} \mathcal{C}_j = \RR^d$. Additionally, assume these convex cones have disjoint interiors and satisfy $\langle u,v\rangle \geq \frac{3}{4}\Vert u \Vert \cdot \Vert v\Vert$ whenever $u,v \in \mathcal{C}_j$. Define the following sets for $j \in \{1,\dots,J\}$:
    \[B_{x,j} = \left(\left(B(x,R)\setminus\bar{B}(x,2\sqrt{z})\right) \cap (\mathcal{C}_j + x)\right) \times (0,z].\]
    If $\varphi \cap B_{x,j} \neq \emptyset$ for all $j \in \{1,\dots,J\}$ then $C((x,h),\varphi) \subset B(x,R)$.
\end{lemma}

In the theorem below we show that $T_n$ almost surely converges pointwise to a function $T$. Note that the integral in the definition of $T$ is finite by Lemma \ref{pF_moments_lemma}.

\begin{theorem}\label{criterion_func_asymptotics}
Let $\lambda > 0$ and $c \in \RR^d$, then with probability one we have :
    \[\lim_{n \to \infty}T_n(\lambda,c) = T(\lambda,c) := \begin{cases} \Vert c - c_0 \Vert^2 + \int_{\RR^d} \Vert y \Vert^2 p_F(y)\mathrm{d}y & \text{if } \lambda=\lambda_0\\
\infty & \text{otherwise}.
\end{cases}\]
\end{theorem}

The proof of Theorem \ref{criterion_func_asymptotics} is given in section \ref{section_proofs_inversion}. In the remainder of this section we derive a few additional limit results, which we need later.

\begin{lemma}\label{lemma_common_limits}
    For $n \in \NN$ define:
    \begin{align}
        D_n &= \frac{1}{\nu_d(W_n)}\sum_{(x,h) \in \eta} \mathds{1}_{W_n}(x)v_{x,h} \label{def_dn_mean_volume}\\
        B_n &= \frac{1}{\nu_d(W_n)}\sum_{(x,h) \in \eta} \mathds{1}_{W_n}(x)v_{x,h} \left(c_{x,h}-x\right). \label{def_bn_centered_centroids}
    \end{align}
    Then, with probability one: $\lim_{n\to \infty}D_n = 1$, and $\lim_{n\to\infty}B_n=0$
\end{lemma}
\begin{proof}
The fact that $\lim_{n\to \infty}D_n = 1$ almost surely, was shown in the proof of Lemma 5 in \cite{vdjagt2025b}, where the authors denoted $D_n$ by $\tilde{F}_n^V(\infty)$. Let $c \in \RR^d$, we may write:
    \begin{align*}
        T_n(\lambda_0,c) &= \frac{1}{\nu_d(W_n)}\sum_{(x,h)\in \eta} \mathds{1}_{W_n}(x)\int_{C((x,h),\eta)} \Vert x - c_0 + c - y \Vert^2\mathrm{d}y \\
        \begin{split}
            &= \frac{1}{\nu_d(W_n)}\sum_{(x,h)\in \eta} \mathds{1}_{W_n}(x)\bigg(v_{x,h}\Vert c - c_0 \Vert^2 + 2\left\langle c - c_0, \int_{C((x,h),\eta)} x-y\mathrm{d}y \right\rangle + \\
            &\phantom{=} \ +\int_{C((x,h),\eta)}\Vert x - y \Vert^2\mathrm{d}y \bigg)
        \end{split}\\
        &= D_n \Vert c - c_0 \Vert^2 -2\langle c - c_0,B_n \rangle + T_n(\lambda_0,c_0).
    \end{align*}
    Indeed, note that for any $(x,h) \in \eta^*$:
    \[ \int_{C((x,h),\eta)} y-x\mathrm{d}y = \int_{C((x,h),\eta)-x} y\mathrm{d}y = v_{x,h} c(C((x,h),\eta)-x) = v_{x,h}(c_{x,h}-x).\]
    As a consequence:
    \begin{equation}
        \langle c - c_0,B_n \rangle = \frac{1}{2}\left(D_n \Vert c - c_0 \Vert^2 + T_n(\lambda_0,c_0) -T_n(\lambda_0,c) \right).\label{Bn_limit_proof_eq}
    \end{equation}
    By Theorem \ref{criterion_func_asymptotics}, and since $\lim_{n\to \infty}D_n = 1$ almost surely, we obtain via the continuous mapping theorem the following almost sure limit:
    \begin{align*}
        \lim_{n \to \infty} \langle c - c_0,B_n \rangle &= \frac{1}{2}\left(\Vert c - c_0 \Vert^2 + \int_{\RR^d} \Vert y \Vert^2 p_F(y)\mathrm{d}y - \Vert c - c_0 \Vert^2 - \int_{\RR^d} \Vert y \Vert^2 p_F(y)\mathrm{d}y\right)\\ 
        &= 0.
    \end{align*}
    Keeping in mind that this holds for any $c \in \RR^d$, we may for instance take $c = c_0 + e_j$ where $e_j$ with $j \in \{1,\dots,d\}$ is the $j$-th standard unit basis vector of $\RR^d$. Taking this $c$ we observe that the $j$-th component of $B_n$ converges to 0 almost surely. Because this holds for any $j \in \{1,\dots,d\}$, $\lim_{n\to\infty}B_n=0$ almost surely.
\end{proof}

For any $c \in \RR^d$, let $\bar{\lambda}_n(c) = \argmin_{\lambda}T_n(\lambda, c)$. Note that $\lambda \mapsto T_n(\lambda,c)$ is convex. Recalling that the partial derivatives of $T_n$ were derived in section \ref{section_definition_inversion}, a straightforward computation yields:
\begin{equation}
    \bar{\lambda}_n(c) = \frac{\sum_{(x,h) \in f_0(\eta)} \mathds{1}_{W_n^0}(x) v_{x,h}\left\langle c_{x,h} - c, x \right\rangle}{\sum_{(x,h) \in f_0(\eta)} \mathds{1}_{W_n^0}(x) v_{x,h} \left\Vert x \right\Vert^2}.\label{lambda_c_def}
\end{equation}
Computing the partial second derivative of $T_n$ w.r.t. $\lambda$ yields:
\begin{equation}
    \frac{\partial^2 T_n(\lambda, c)}{\partial \lambda^2} = \frac{2}{\nu_d(W_n)}\sum_{(x,h) \in f_0(\eta)} \mathds{1}_{W_n^0}(x)v_{x,h}\Vert x \Vert^2 =:2M_n. \label{Mn0_def}
\end{equation}
From this, it immediately follows that for any $c \in \RR^d$, $T_n(\blank,c)$ is strongly convex with parameter $2M_n$. Recall that a function $f:\RR^d \to \RR$ is strongly convex with parameter $m > 0$ if for all $x,y \in \RR^d$:
\[f(y) \geq f(x) + \nabla \langle f(x),y-x\rangle + \frac{m}{2}\Vert y - x\Vert^2.\]
Observe that if $x^*$ is the global minimizer of $f$, then:
\begin{equation}
    \Vert y - x^*\Vert^2 \leq \frac{2}{m}\left(f(y) - f(x^*)\right).\label{stronglyconvex_minimizer}
\end{equation}
A sufficient condition for $f:\RR \to \RR$ to be strongly convex with parameter $m > 0$ is that $f$ is twice differentiable and $f''(x) \geq m$ for all $x \in \RR$. We now show that the strong convexity parameter of $T_n(\blank,c)$ diverges to infinity as $n \to \infty$:

\begin{lemma}\label{Mn0_limit_lemma}
    With probability one, $\lim_{n \to \infty} M_n = \infty$. Here, $M_n$ is as in (\ref{Mn0_def}).
\end{lemma}
\begin{proof}
Let $(\Omega,\mathcal{A},\PP)$ be a probability space supporting the Poisson process $\eta$. As in the proof of Theorem \ref{criterion_func_asymptotics} it is sufficient to show that for all $M > 0$ there exists a $\Omega_M \in \mathcal{A}$ with $\PP(\Omega_M)=1$ such that for all $\omega \in \Omega_M$: $\liminf_{n \to \infty}M_n(\omega) > M$. Let $M > 0$. Choose $\Omega_M \in \mathcal{A}$ with $\PP(\Omega_M)=1$ such that for all $\omega \in \Omega_M$: $L(\eta(\omega))$ is a tessellation and $\lim_{n \to \infty}D_n(\omega)=1$. Such a $\Omega_M$ exists by Lemma \ref{lemma_common_limits} and the fact that Poisson-Laguerre tessellations are well-defined. Let $\omega \in \Omega_M$, then we may write $M_n(\omega)$ as:
{\allowdisplaybreaks
\begin{align}
        M_n(\omega) &= \frac{1}{\nu_d(W_n)}\sum_{(x,h) \in \eta(\blank;\omega)} \mathds{1}_{W_n}(x)v_{x,h}(\omega)\left\Vert \frac{x}{\lambda_0} - \frac{c_0}{\lambda_0} \right\Vert^2 \nonumber \\
        &= \frac{1}{\nu_d(W_n)}\sum_{(x,h) \in \eta(\blank;\omega)} \mathds{1}_{W_n}(x)v_{x,h}(\omega)\left\Vert \frac{x}{\lambda_0} - \frac{c_0}{\lambda_0} \right\Vert^2 \mathds{1}\left\{\left\Vert \frac{x}{\lambda_0} - \frac{c_0}{\lambda_0} \right\Vert\leq \sqrt{M} \right\} + \label{Mn0_term1}\\
        &\phantom{=} \ + \frac{1}{\nu_d(W_n)}\sum_{(x,h) \in \eta(\blank;\omega)} \mathds{1}_{W_n}(x)v_{x,h}(\omega)\left\Vert \frac{x}{\lambda_0} - \frac{c_0}{\lambda_0} \right\Vert^2 \mathds{1}\left\{\left\Vert \frac{x}{\lambda_0} - \frac{c_0}{\lambda_0} \right\Vert > \sqrt{M} \right\} \label{Mn0_term2}
    \end{align}
}%
    We consider (\ref{Mn0_term1}) and (\ref{Mn0_term2}) separately. Note that the following expression is an upper bound for (\ref{Mn0_term1}):
    \begin{equation}
      \frac{M}{\nu_d(W_n)}\left(\sum_{(x,h) \in \eta(\blank;\omega)} \mathds{1}_{W_n \cap \bar{B}(c_0,\lambda_0 M)}(x)v_{x,h}(\omega)\right).\label{Mn0_term1_bound}
    \end{equation}
    For sufficiently large $n$, $W_n \cap \bar{B}(c_0,\lambda_0 M) = \bar{B}(c_0,\lambda_0 M)$. Hence, for such $n$, the term in brackets in (\ref{Mn0_term1_bound}) does not depend on $n$. Because the tessellation $L(\eta(\omega))$ is locally finite and has cells of finite volume, this term is finite. Because $\nu_d(W_n) \to \infty$ as $n \to \infty$ we see that (\ref{Mn0_term1}) vanishes. For (\ref{Mn0_term2}) we obtain the following lower bound:
    \begin{align}
        &\frac{1}{\nu_d(W_n)}\sum_{(x,h) \in \eta(\blank;\omega)} \mathds{1}_{W_n}(x)v_{x,h}(\omega)\left\Vert \frac{x}{\lambda_0} - \frac{c_0}{\lambda_0} \right\Vert^2 \mathds{1}\left\{\left\Vert \frac{x}{\lambda_0} - \frac{c_0}{\lambda_0} \right\Vert > \sqrt{M} \right\} \nonumber \\
        &\quad> D_n(\omega)M - \frac{1}{\nu_d(W_n)}\sum_{(x,h) \in \eta(\blank;\omega)} \mathds{1}_{W_n}(x)v_{x,h}(\omega)M\mathds{1}\left\{\left\Vert \frac{x}{\lambda_0} - \frac{c_0}{\lambda_0} \right\Vert \leq \sqrt{M} \right\} .\label{Mn0_term2_bound}
    \end{align}
    By the choice of $\Omega_M$, the first term of (\ref{Mn0_term2_bound}) converges to $M$ as $n \to \infty$, while the second term of (\ref{Mn0_term2_bound}) vanishes, as $n \to \infty$. The fact that the second term of (\ref{Mn0_term2_bound}) vanishes can be shown via the same argument as used for (\ref{Mn0_term1}). Hence: $\liminf_{n \to \infty} M_n(\omega) > M$. 
\end{proof}

As a consequence we find that a strongly consistent estimator for $\lambda_0$ can be obtained even if we do not optimize for $c$. 

\begin{corollary}\label{corr_consistentcy_lambda_c}
    Let $c \in \RR^d$, then with probability one: $\lim_{n \to \infty} \bar{\lambda}_n(c) = \lambda_0$.
\end{corollary}

\begin{proof}
Applying (\ref{stronglyconvex_minimizer}) to the function $T_n(\blank,c)$, and taking into account that $T_n(\blank,c)$ is non-negative, we obtain:
\begin{equation}
    \left\Vert \lambda_0 - \bar{\lambda}_n(c) \right\Vert^2 \leq \frac{1}{M_n}\left(T_n(\lambda_0,c) - T_n(\bar{\lambda}_n(c),c)\right) \leq \frac{T_n(\lambda_0,c)}{M_n}.\label{lambda_c_minimizer}
\end{equation}
By Lemma (\ref{Mn0_limit_lemma}), Theorem \ref{criterion_func_asymptotics} and the continuous mapping theorem we obtain that the RHS of (\ref{lambda_c_minimizer}) vanishes almost surely, which yields the result.
\end{proof}

\section{Consistency of the inversion procedure}\label{inversion_sufficient_conditions_section}

In the previous section we obtained various limit results, which we will need in this section for proving consistency of the inversion procedure. That is, we will show that under reasonable conditions the proposed estimator $(\hat{\lambda}_n,\hat{c}_n)$ for $(\lambda_0,c_0)$ is consistent. For $n \in \NN$ define the following random vector:
\begin{equation}
    A_n = \frac{1}{\nu_d(W_n)}\sum_{(x,h) \in \eta}\mathds{1}_{W_n}(x)\nu_d(C(x,h),\eta)x. \label{def_an_mean_generator}
\end{equation}
That is, $A_n$ may be seen as a weighted average of the generators corresponding to non-empty cells, the weights being the volumes of the cells. One might argue that for a proper weighted average we should divide $A_n$ by $D_n$, with $D_n$ as in (\ref{def_dn_mean_volume}). However, by Lemma \ref{lemma_common_limits} we know that $D_n$ converges to $1$ almost surely as $n \to \infty$. Hence, when studying the behavior of $A_n$ as $n \to \infty$, it is not needed to divide $A_n$ by $D_n$ because this normalization will not lead to a different limit of $A_n$, if it exists in some sense of stochastic convergence. As a first observation, the expected value of $A_n$ is the centroid of $W_n$:
\begin{align*}
    \EE\left(A_n\right) = \frac{1}{\nu_d(W_n)}\int_{W_n}\int_0^\infty \EE(\nu_d(C(x,h),\eta))\mathrm{d}F(h)x\mathrm{d}x = \frac{1}{\nu_d(W_n)}\int_{W_n}x\mathrm{d}x = c(W_n).
\end{align*}
This follows from the Mecke equation (Theorem 4.1. in \cite{Last2018}) and (\ref{mean_cell_volume1}). As it turns out, the behavior of $A_n$ is essential for obtaining a consistency result for $(\hat{\lambda}_n,\hat{c}_n)$, as is highlighted by the theorem below.

\begin{theorem}[Consistent inversion]\label{thm_consistent_inversion}\phantom{text}
\begin{enumerate}
    \item If the sequence $(A_n)_{n \geq 1}$ is uniformly tight, then:
    \[\lim_{n \to \infty} \left(\hat{\lambda}_n,\hat{c}_n\right) \overset{\PP}{=} (\lambda_0,c_0).\]
    \item If the random variable $\sup_{n \geq 1}\Vert A_n \Vert$ is almost surely finite, then:
    \[\lim_{n \to \infty}\left(\hat{\lambda}_n,\hat{c}_n\right) \overset{\text{a.s.}}{=} (\lambda_0,c_0).\]
\end{enumerate}
\end{theorem}

\begin{proof}
Let $A_n$, $B_n$ and $D_n$ be as in (\ref{def_an_mean_generator}), (\ref{def_bn_centered_centroids}) and (\ref{def_dn_mean_volume}) respectively. For $n \in \NN$ observe that:
\[\frac{1}{\nu_d(W_n)}\sum_{(x,h) \in \eta} \mathds{1}_{W_n}(x)v_{x,h} c_{x,h} = B_n + A_n.\]
Let $\bar{\lambda}_n(0)$ and $M_n$ be as in (\ref{lambda_c_def}) and (\ref{Mn0_def}) respectively. We may write $\hat{\lambda}_n$ (recall (\ref{cn_lambdan_formula})) as follows:
{\allowdisplaybreaks
\begin{align*}
        \hat{\lambda}_n &= \frac{D_n \frac{1}{\nu_d(W_n)}\sum_{(x,h) \in f_0(\eta)} \mathds{1}_{W_n^0}(x)v_{x,h}\langle c_{x,h},x\rangle - \left\langle B_n+A_n,A_n \right\rangle}{D_n \left(\frac{1}{\nu_d(W_n)}\sum_{(x,h) \in f_0(\eta)} \mathds{1}_{W_n^0}(x)v_{x,h}\Vert x \Vert^2 \right) - \left\Vert A_n \right\Vert^2} \\
        \begin{split}
            &= \frac{D_n \frac{1}{\nu_d(W_n)}\sum_{(x,h) \in f_0(\eta)} \mathds{1}_{W_n^0}(x)v_{x,h}\langle c_{x,h},x\rangle - \left\Vert A_n \right\Vert^2 -  \left\langle B_n,A_n \right\rangle}{D_n \left(\frac{1}{\nu_d(W_n)}\sum_{(x,h) \in f_0(\eta)} \mathds{1}_{W_n^0}(x)v_{x,h}\Vert x \Vert^2 \right) }\cdot \\
            &\phantom{=} \ \cdot\frac{D_n \left(\frac{1}{\nu_d(W_n)}\sum_{(x,h) \in f_0(\eta)} \mathds{1}_{W_n^0}(x)v_{x,h}\Vert x \Vert^2 \right) }{D_n \left(\frac{1}{\nu_d(W_n)}\sum_{(x,h) \in f_0(\eta)} \mathds{1}_{W_n^0}(x)v_{x,h}\Vert x \Vert^2 \right) - \left\Vert A_n \right\Vert^2}
        \end{split}\\
        &=\left( \bar{\lambda}_n(0) - \frac{\left\Vert A_n \right\Vert^2 +  \left\langle B_n,A_n \right\rangle}{D_n M_n}\right) \frac{1}{1 - \frac{\Vert A_n \Vert^2}{D_nM_n}}
    \end{align*}
}%
By Corollary \ref{corr_consistentcy_lambda_c} we have $\lim_{n \to \infty}\bar{\lambda}_n(0)=\lambda_0$ almost surely. Suppose condition 1 of the theorem is satisfied. By the uniform tightness of $A_n$, and since $\lim_{n \to \infty}B_n=0$ almost surely by Lemma \ref{lemma_common_limits}, we obtain that $\Vert A_n \Vert^2 +\langle B_n,A_n\rangle$ is uniformly tight. Because $\lim_{n \to \infty} M_n = \infty$ almost surely this yields:
\begin{equation}
    \lim_{n \to \infty}\frac{\left\Vert A_n \right\Vert^2 +  \left\langle B_n,A_n \right\rangle}{D_n M_n} \overset{\PP}{=}0 \text{ \quad and \quad } \lim_{n \to \infty}\frac{1}{1 - \frac{\Vert A_n \Vert^2}{D_nM_n}}\overset{\PP}{=}1.\label{consistency_lambdan_intermediate_limits}
\end{equation}
Here we also used the rule which is often written as $o_P(1)O_P(1) = o_P(1)$ in stochastic o notation. Hence, via the continuous mapping theorem we obtain $\lim_{n \to \infty}\hat{\lambda}_n = \lambda_0$ in probability. If condition 2 of the theorem is satisfied then the limits in (\ref{consistency_lambdan_intermediate_limits}) become almost sure limits, yielding $\lim_{n \to \infty}\hat{\lambda}_n = \lambda_0$ almost surely. We now consider $\hat{c}_n$ (recall (\ref{cn_lambdan_formula})), which may be written as follows:
    \begin{align}
        \hat{c}_n &= \frac{\sum_{(x,h) \in f_0(\eta)} \mathds{1}_{W_n^0}(x)v_{x,h} c_{x,h}}{\sum_{(x,h) \in f_0(\eta)} \mathds{1}_{W_n^0}(x)v_{x,h}} - \hat{\lambda}_n \frac{\sum_{(x,h) \in f_0(\eta)} \mathds{1}_{W_n^0}(x)v_{x,h} x}{\sum_{(x,h) \in f_0(\eta)} \mathds{1}_{W_n^0}(x)v_{x,h}} \nonumber \\
        &= \frac{\sum_{(x,h) \in \eta} \mathds{1}_{W_n}(x)v_{x,h} c_{x,h}}{\sum_{(x,h) \in \eta} \mathds{1}_{W_n}(x)v_{x,h}} - \hat{\lambda}_n \frac{\sum_{(x,h) \in \eta} \mathds{1}_{W_n}(x)v_{x,h}\left(\frac{x}{\lambda_0} -\frac{c_0}{\lambda_0} \right)}{\sum_{(x,h) \in \eta} \mathds{1}_{W_n}(x)v_{x,h}}\nonumber\\
        &= \frac{\sum_{(x,h) \in \eta} \mathds{1}_{W_n}(x)v_{x,h} \left(c_{x,h}-x\right)}{\sum_{(x,h) \in \eta} \mathds{1}_{W_n}(x)v_{x,h}} + \frac{1}{\lambda_0}(\lambda_0 - \hat{\lambda}_n)\frac{\sum_{(x,h) \in \eta} \mathds{1}_{W_n}(x)v_{x,h}x}{\sum_{(x,h) \in \eta} \mathds{1}_{W_n}(x)v_{x,h}} + \frac{\hat{\lambda}_n}{\lambda_0}c_0.\nonumber\\
        &= \frac{B_n}{D_n} + \frac{1}{\lambda_0}(\lambda_0 - \hat{\lambda}_n)\frac{A_n}{D_n}+\frac{\hat{\lambda}_n}{\lambda_0}c_0.\label{hat_cn_simplified_expression}
    \end{align}
    By Lemma \ref{lemma_common_limits}, $\lim_{n \to \infty}B_n/D_n=0$ almost surely. If condition 1 of the theorem is satisfied, then $\lim_{n \to \infty}\hat{\lambda}_n=\lambda_0$ in probability. Hence, 
    \[\lim_{n \to \infty}\frac{(\lambda_0 - \hat{\lambda}_n)}{\lambda_0D_n} \overset{\PP}{=}0,\]
    here was also used that $\lim_{n \to \infty}D_n=1$ almost surely by Lemma \ref{lemma_common_limits}. As a consequence, the second term of (\ref{hat_cn_simplified_expression}) vanishes when taking the limit in probability ($o_P(1)O_P(1) = o_P(1)$), since $A_n$ is uniformly tight. Via the continuous mapping theorem we obtain $\lim_{n \to \infty} \hat{c}_n=c_0$ in probability. If condition 2 of the theorem is satisfied, then $\lim_{n \to \infty}\hat{\lambda}_n=\lambda_0$ almost surely. Via similar arguments we obtain $\lim_{n \to \infty} \hat{c}_n=c_0$ almost surely.
\end{proof}

It seems necessary that some conditions on $(c(W_n))_{n\geq 1}$ need to be imposed for any of the two conditions in Theorem \ref{thm_consistent_inversion} to hold. As is commonly done we may for instance choose a convex averaging sequence $(W_n)_{n\geq 1}$ which is centered at the origin. That is, we may take $W_n$ such that $c(W_n)=0$ for all $n \in \NN$. Obviously, if $A_n$ then converges almost surely to 0, this would imply condition 2 in Theorem \ref{thm_consistent_inversion}. While we do not know if $A_n$ converges in general, we will now highlight via an example why the almost sure convergence of $A_n$ to 0 should not necessarily be expected, especially if $d=2$. 

\begin{example}\label{example_mean_poisson_points}
    The vector $A_n$ may essentially be seen as a weighted average of Poisson points, the weights being the volumes of the Laguerre cells. Let us now consider what happens if we consider the arithmetic mean of Poisson points instead. The results in this example are given without proof, as they can be obtained via relatively straightforward computations. Suppose $\psi$ is a homogeneous Poisson process on $\RR^d$ with intensity 1. For convenience, we take $W_n = [-\frac{1}{2}n^{1/d},\frac{1}{2}n^{1/d}]^d$. That is, $W_n$ is the centered cube of volume $\nu_d(W_n)=n$. Then, we define:
\[\bar{A}_n = \frac{1}{\nu_d(W_n)}\sum_{x \in \psi}\mathds{1}_{W_n}(x)x.\]
Of course, for an actual arithmetic mean we should normalize with $\psi(W_n)$ instead of $\nu_d(W_n)$. However, by the spatial ergodic theorem $\lim_{n\to\infty}\psi(W_n)/\nu_d(W_n)=1$ with probability one. So this normalization will not lead to a different limit of $\bar{A}_n$, if it exists in some sense of stochastic convergence. It can be shown that $\bar{A}_n$ has the following limit in distribution:
\begin{align*}
    \lim_{n \to \infty} \bar{A}_n = \begin{cases}
        \mathcal{N}\left(0,\frac{1}{12}I_2\right) & \text{ if } d=2\\
        0 & \text{ if } d \geq3.
    \end{cases}
\end{align*}
Here, $\mathcal{N}(0,\smash{\frac{1}{12}}I_2)$ denotes the multivariate normal distribution on $\RR^2$ with mean $0$ and covariance matrix $\smash{\frac{1}{12}}I_2$, $I_2$ being the $2\times 2$ identity matrix. Recall that convergence in distribution to a constant implies convergence in probability to the same constant. So if $d \geq 3$ the convergence also holds in probability. It can be shown that the $L^2$-norm of $\bar{A}_n$ is given by:
\[\EE\left(\Vert\bar{A}_n \Vert^2 \right) =\frac{d}{12n^{1 - \frac{2}{d}}}.\]
The $L^2$-norm highlights that the rate of convergence of $\bar{A}_n$ to $0$ in $L^2$ for $d \geq 3$ is faster in higher dimensions.
\end{example}

In the remainder of this section we will show that under reasonable conditions the $L^2$-norm of $A_n$ is also of order $O(n^{-1+2/d})$. As a consequence, under those conditions, $(A_n)_{n\geq 1}$ is uniformly tight by Markov's inequality. In order to compute a bound for this $L^2$-norm we need several results which are stated below. The proofs of these statements are given in section \ref{section_proofs_inversion}.

\begin{proposition}\label{prop_second_moment_bound}
    Let $x \in \RR^d$, $h \geq 0$ and $p \in \NN$. We have the following inequalities:
    \begin{align*}
        \EE\left(\nu_d(C((x,h),\eta))^p\right) \leq pd\kappa_d^p\int_0^\infty \exp\left(-\kappa_d\int_0^{r^2 + h}\left(r^2 + h -t \right)^\frac{d}{2}\mathrm{d}F(t)\right)r^{pd-1}\mathrm{d}r \leq \frac{p!}{F(h)^p}.
    \end{align*}
    Here, we set $p!/F(h)^p = \infty$ if $F(h) = 0$.
\end{proposition}

For proving Proposition \ref{prop_second_moment_bound} we use a similar technique as was used in Proposition 2.3 in \cite{Olsbo2007} to derive a bound on the second moment of the volume of a typical Poisson-Voronoi cell. While the second inequality in Proposition \ref{prop_second_moment_bound} provides a finite bound for large values of $h$ it is not yet clear whether the $p$-th moment is finite for all $h \geq 0$. We show this in the lemma below.

\begin{lemma}\label{lemma_individual_second_moment_finite}
    Let $p \in \NN$. There exists a constant $ 0 < \alpha_p< \infty$ such that for all $x \in \RR^d$ and $h \geq 0$, we have $\EE\left(\nu_d(C((x,h),\eta))^p\right) < \alpha_p$.
\end{lemma}

\begin{theorem}\label{thm_integrated_second_moment}
    Let $p \in \NN$. There exists a constant $0 < \beta_p < \infty$ such that for all $x \in \RR^d$:
    \[\int_0^\infty \EE\left(\nu_d(C((x,h),\eta))^p\right) \mathrm{d}F(h) \leq \beta_p.\]
\end{theorem}

\begin{lemma}\label{lemma_bound_product_cell_volumes}
    Let $x_1,x_2 \in \RR^d$, then:
    \[\int_0^\infty \int_0^\infty \EE\left(\nu_d(C((x_1,h_1),\eta+\delta_{(x_2,h_2)}))\nu_d(C((x_2,h_2),\eta+\delta_{(x_1,h_1)})) \right)\mathrm{d}F(h_1)\mathrm{d}F(h_2) \leq 4.\]
\end{lemma}

We set $W_n = n^\frac{1}{d}W$ for $n \in \NN$, where $W \subset \RR^d$ is a convex body satisfying $\nu_d(W) = 1$ and $c(W) = 0$. Via a direct computation we may write:
    \begin{align}
\begin{split}
    &\EE\left( \left\Vert A_n \right\Vert^2 \right) = \frac{1}{\nu_d(W_n)^2}\EE\left(\sum_{(x,h) \in \eta}\mathds{1}_{W_n}(x)\nu_d(C(x,h),\eta)^2\Vert x\Vert^2 \right) + \\
    & \ +\frac{1}{\nu_d(W_n)^2}\EE\left(\sum_{(x_1,h_1),(x_2,h_2) \in (\eta)_{\neq}^2}\mathds{1}_{W_n}(x_1,x_2)\nu_d(C(x_1,h_1),\eta)\nu_d(C(x_2,h_2),\eta)\langle x_1,x_2\rangle \right).
\end{split} \label{expectation_an_squared}
\end{align}
Here, $(\eta)_{\neq}^2$ denotes the set of all distinct pairs of points of $\eta$. That is, if $(x_1,h_1),(x_2,h_2) \in (\eta)_{\neq}^2$, then $(x_1,h_1) \neq (x_2,h_2)$.  Via the Mecke equation and Theorem \ref{thm_integrated_second_moment} the first term of (\ref{expectation_an_squared}) is bounded from above by:
\begin{align}
    \frac{1}{\nu_d(W_n)^2}\EE\left(\sum_{(x,h) \in \eta}\mathds{1}_{W_n}(x)\nu_d(C(x,h),\eta)^2\Vert x\Vert^2 \right) &= \frac{1}{n^2} \int_{W_n} \int_0^\infty \EE\left(\nu_d(C(x,h),\eta)^2 \right)\mathrm{d}F(h) \Vert x \Vert^2 \mathrm{d}x\nonumber\\
    &\leq \frac{\beta_2}{n^2} \int_{W_n} \Vert x \Vert^2 \mathrm{d}x\nonumber\\
    &\leq \frac{\beta_2}{n^2} \int_{W_n} \mathrm{diam}(W_n)^2 \mathrm{d}x\nonumber\\
    &= \frac{\beta_2 \mathrm{diam}(W)^2}{n^{1-\frac{2}{d}}}. \label{an_term1_l2_bound}
\end{align}
Here we also used: $\mathrm{diam}(W_n) = n^{1/d}\mathrm{diam}(W)$. Notice that this upper bound is similar to the $L^2$-norm of $\bar{A}_n$ as in example \ref{example_mean_poisson_points}. In the case of $\bar{A}_n$, the cross-terms vanish, due to independence. If this is also the case for $A_n$, this would mean the second expectation in (\ref{expectation_an_squared}) is equal to zero. Define the function $\phi : \RR^d \times \RR^d \to [0,\infty)$ via:
\begin{align*}
    \phi(x_1,x_2) = \int_0^\infty \int_0^\infty \EE\left(\nu_d(C((x_1,h_1),\eta+\delta_{(x_2,h_2)}))\nu_d(C((x_2,h_2),\eta+\delta_{(x_1,h_1)}) \right)\mathrm{d}F(h_1)\mathrm{d}F(h_2).
\end{align*}
Then, via the multivariate Mecke equation (Theorem 4.4 in \cite{Last2018}) the second term in (\ref{expectation_an_squared}) is given by:
\begin{equation}
    \frac{1}{n^2}\int_{W_n}\int_{W_n}\phi(x_1,x_2)\langle x_1, x_2\rangle \mathrm{d}x_1\mathrm{d}x_2.\label{integral_an_cross_term}
\end{equation}
By Lemma \ref{lemma_bound_product_cell_volumes} we know that $|\phi(x_1,x_2)| \leq 4$. Taking this into account and by applying the Cauchy-Schwarz inequality, we obtain the following upper bound for (\ref{integral_an_cross_term}):
\[\left|\frac{1}{n^2}\int_{W_n}\int_{W_n}\phi(x_1,x_2)\langle x_1, x_2\rangle \mathrm{d}x_1\mathrm{d}x_2\right| \leq \frac{4}{n^2}\int_{W_n}\int_{W_n}\left|\langle x_1, x_2\rangle\right| \mathrm{d}x_1\mathrm{d}x_2 \leq  \mathrm{diam(W)^2} 4n^\frac{2}{d}.\]
Combining this expression with (\ref{an_term1_l2_bound}) yields:
\begin{equation}
    \EE\left(\Vert A_n\Vert^2 \right) \leq \frac{\beta_2 \mathrm{diam}(W)^2}{n^{1-\frac{2}{d}}} + \mathrm{diam(W)^2} 4n^\frac{2}{d}. \label{an_simple_l2_bound}
\end{equation}
Unfortunately, the RHS of (\ref{an_simple_l2_bound}) is not uniformly bounded in $n$. It does yield $\EE(\Vert A_n\Vert^2) < \infty$ for each $n \in \NN$ which we will need later. Currently, it seems that deriving a tight upper bound for $\EE\left(\Vert A_n\Vert^2 \right)$ is hard in general. However, under certain assumptions on $\mathbb{F}$ it becomes possible to obtain a bound of order $O(n^{-1+2/d})$. To prove this, we require the notion of stabilization, which is often used in stochastic geometry for deriving laws of large numbers and central limit theorems, we refer to \cite{Schreiber2010} for an overview. In the context of Poisson-Laguerre tessellations stabilization techniques were used in \cite{Flimmel2020} for proving asymptotic normality of estimators of geometric characteristics of cells. 

\begin{theorem}\label{stabilization_theorem}[Propositions 3.1 and 3.2 in \cite{Flimmel2020}]
    Suppose $\mathbb{F}$ is concentrated on $(0,M)$ for some $M > 0$, and let $(x,h) \in \RR^d \times (0,M)$. There exists a random variable $R_x(\eta) > 0$ (the radius of stabilization) such that:
    \begin{enumerate}
        \item With probability one, for all $(x',h') \in \RR^d \times (0,M)$, with $\Vert x - x'\Vert > R_x(\eta)$:
        \[C((x,h),\eta) = C\left((x,h),\eta+\delta_{(x',h')}\right).\]
        \item With probability one:
        \[C((x,h),\eta) \subset \bar{B}(x,R_x(\eta)).\]
        \item The random variable $R_x(\eta)$ has exponentially decaying tails. That is, there exists constants $c_1,c_2 > 0$ such that for all $r \geq 0$:
        \[\PP(R_x(\eta) > r)<c_1e^{-c_2r}.\]
        \item The distribution of $R_x(\eta)$ does not depend on $x$. Additionally, $R_x(\eta+\delta_x) = R_x(\eta)$ and $R_x(\eta+\delta_{x'}) \leq R_x(\eta)$ for all $x' \in \RR^d$.
    \end{enumerate}
\end{theorem}

Additionally, we recall the Poincar\'e inequality, which is useful for deriving bounds of the variance of so-called Poisson functionals, if one can control the influence of inserting an additional point into the point process.

\begin{lemma}[Poincar\'e inequality, Theorem 18.7 in \cite{Last2018}]\label{poincare_inequality_lemma}
    Suppose $\eta$ is a Poisson process on the measurable space $(\mathbb{X},\mathcal{X})$ with $\sigma$-finite intensity measure $\Lambda$. Suppose $f:\mathbf{N}(\mathbb{X}) \to \RR$ is such that $\EE(f(\eta)^2)<\infty$, then:
    \[\mathrm{Var}(f(\eta)) \leq \int_\mathbb{X}\EE\left((D_x f(\eta))^2\right)\Lambda(\mathrm{d}x),\]
    with $D_xf(\eta) = f(\eta+\delta_x)-f(\eta)$, for $x \in \mathbb{X}$.
\end{lemma}

Proposition \ref{poincare_inequality_lemma} and Theorem \ref{stabilization_theorem} may be used to prove the theorem below. The proof of Theorem \ref{thm_an_L2_bound} is given in section \ref{section_proofs_inversion}.

\begin{theorem}\label{thm_an_L2_bound}
    Let $A_n$ be as in (\ref{def_an_mean_generator}) with $W_n = n^\frac{1}{d}W$ for $n \in \NN$, where $W \subset \RR^d$ is a convex body satisfying $\nu_d(W) = 1$ and $c(W) = 0$. If $\mathbb{F}$ is concentrated on $(0,M)$ for some $M > 0$, then there exists a constant $0 < \gamma < \infty$ such that for all $n \in \NN$:
    \[\EE\left( \left\Vert A_n \right\Vert^2 \right) \leq \frac{\gamma}{n^{1-\frac{2}{d}}}.\]
    As a consequence, $\lim_{n \to \infty} A_n = 0$ in $L^2$ for $d \geq 3$. If $d\geq 2$, the sequence $(A_n)_{n\geq 1}$ is uniformly tight.
\end{theorem}

Theorem \ref{thm_an_L2_bound} guarantees consistency of the inversion procedure for a large class of choices for $\mathbb{F}$. Whether the inversion procedure is consistent for all locally finite measures $\mathbb{F}$ on $(0,\infty)$ is an open problem.

\begin{remark}
    In general, if one does not wish to assume a uniform upper bound for the weights, it is not clear whether there is still a way to apply stabilization techniques. In the recent paper \cite{Bhattacharjee2025} a different technique, called region-stabilization, is used in the context of Poisson-Laguerre tessellations for deriving central limit theorems. Here, various parametric models are considered with unbounded weights. Perhaps similar techniques can be used to derive a bound for $\EE(\Vert A_n\Vert^2)$ without assuming an upper bound for the weights.
\end{remark}


\section{Simulations}\label{section_simulations_inversion}

In this section we perform various simulations to empirically study the behavior of the estimator $(\hat{\lambda}_n,\hat{c}_n)$. We also consider a variant of $(\hat{\lambda}_n,\hat{c}_n)$, which takes into account edge effects. By computing these estimators, we effectively invert the observed Poisson-Laguerre tessellations. Keeping in mind the main motivation for inverting Poisson-Laguerre tessellations, we then study how well the function $F$ can be estimated if the original weighted generators are unknown. That is, first we observe a Poisson-Laguerre tessellation through a bounded window and apply the inversion procedure. Then, we compute a variant of the estimator $\hat{F}_n^0$ for $F$, as defined in Definition 3 in \cite{vdjagt2025b}, which is based on the weighted generator point obtained via the inversion procedure ($\hat{\eta}_n^*$) instead of the actual weighted generators ($\eta^*$). For all simulations we simulate planar ($d=2$) Poisson-Laguerre tessellations.

\subsection{Estimation of \texorpdfstring{$\lambda_0$}{λ₀} and \texorpdfstring{$c_0$}{c₀}}
Recall the definition of $(\hat{\lambda}_n,\hat{c}_n)$ as a solution to the weighted least squares problem in (\ref{estimators_wls_def}). Then, $(\hat{\lambda}_n,\hat{c}_n)$ minimizes a function which cannot actually be computed in practice. After all, when a Poisson-Laguerre tessellation is observed through a window $W_n$, the cells at the boundary of the window are only partially observed and therefore their volumes and centroids cannot be computed. Therefore, we also consider the following estimator for $(\lambda_0,c_0)$:
\[(\bar{\lambda}_n, \bar{c}_n) = \argmin_{(\lambda, c)\in \RR\times\RR^d}\frac{1}{\nu_d(W_n)}\sum_{(x,h) \in f_0(\eta)}\mathds{1}{\{C((x,h),f_0(\eta)) \subset W_n\}}v_{x,h}\left\Vert \lambda x + c - c_{x,h} \right\Vert^2.\]
Because $(\bar{\lambda}_n, \bar{c}_n)$ depends on less information compared to $(\hat{\lambda}_n,\hat{c}_n)$, it is to be expected that its performance will be slightly worse. Loosely speaking, $(\bar{\lambda}_n, \bar{c}_n)$ is based on information contained in $W_n$ and we expect $(\bar{\lambda}_n, \bar{c}_n)$ to behave like $(\hat{\lambda}_n,\hat{c}_n)$ if $(\hat{\lambda}_n,\hat{c}_n)$ is computed based on a smaller observation window $W_n' \subset W_n$ instead of $W_n$. While $(\hat{\lambda}_n,\hat{c}_n)$ cannot be computed in practice, this is not an issue in a simulation setting, where we can effectively also observe cells outside of $W_n$.

\begin{table}[b!]
    \centering
    \begin{tabular}
    {
  r
  S[table-format=1.7,round-precision=3]
>{{{\lp}}} 
S[round-precision=2, table-format = 1.7,table-space-text-pre=\lp]
@{,\,} 
S[round-precision=2, table-format = 1.5,table-space-text-post=\rp]
<{{{\rp}}} 
    S[table-format=1.5, round-precision=3]
>{{{\lp}}} 
S[round-precision=2,table-format = 1.4,table-space-text-pre=\lp]
@{,\,} 
S[round-precision=2,table-format = 1.4,table-space-text-post=\rp]
<{{{\rp}}} 
@{}l@{}
}
\toprule
\multicolumn{1}{c}{} & \multicolumn{3}{c}{$|\bar{\lambda}_n - \lambda_0|$} & \multicolumn{3}{c}{$\Vert \bar{c}_n -c_0\Vert$} & \\
\cmidrule(r){1-1}\cmidrule(lr){2-4}\cmidrule(l){5-8}
  \multicolumn{1}{c}{$P_n$} &        {mean} &    \multicolumn{2}{c}{(2.5\%, 97.5\%)} & {mean} &    \multicolumn{2}{c}{(2.5\%, 97.5\%)} & \\
\cmidrule(r){1-1}\cmidrule(lr){2-4}\cmidrule(l){5-8}
500 & 0.00076635 & 0.00006581 & 0.00152634 & 0.02095217 & 0.00355202 & 0.04220091 &   \\
1000 & 0.00033413 & 0.00001718 & 0.00073036 & 0.01195117 & 0.00201534 & 0.02407493 &   \\
2000 & 0.00014298 & 0.00000959 & 0.00029023 & 0.00674849 & 0.00100166 & 0.01302620 &   \\
5000 & 0.00006172 & 0.00002466 & 0.00010711 & 0.00432652 & 0.00134680 & 0.00768895 &   \\
\bottomrule
\end{tabular}
\caption{Estimates of $\lambda_0$ and $c_0$ obtained with edge correction, the underlying $F$ is given by (\ref{true_F_option1_inversion}).}\label{table_edgecorrected_uniform}
\vspace{5pt}
\begin{tabular}
    {
  r
  S[table-format=1.7,round-precision=3]
>{{{\lp}}} 
S[round-precision=2, table-format = 1.7,table-space-text-pre=\lp]
@{,\,} 
S[round-precision=2, table-format = 1.6,table-space-text-post=\rp]
<{{{\rp}}} 
    S[table-format=1.5, round-precision=3]
>{{{\lp}}} 
S[round-precision=2,table-format = 1.5,table-space-text-pre=\lp]
@{,\,} 
S[round-precision=2,table-format = 1.4,table-space-text-post=\rp]
<{{{\rp}}} 
@{}l@{}
}
\toprule
\multicolumn{1}{c}{} & \multicolumn{3}{c}{$|\hat{\lambda}_n - \lambda_0|$} & \multicolumn{3}{c}{$\Vert \hat{c}_n -c_0\Vert$} & \\
\cmidrule(r){1-1}\cmidrule(lr){2-4}\cmidrule(l){5-8}
  \multicolumn{1}{c}{$P_n$} &        {mean} &    \multicolumn{2}{c}{(2.5\%, 97.5\%)} & {mean} &    \multicolumn{2}{c}{(2.5\%, 97.5\%)} & \\
\cmidrule(r){1-1}\cmidrule(lr){2-4}\cmidrule(l){5-8}
500 & 0.00038558 & 0.00002616 & 0.00119095 & 0.01126269 & 0.00091320 & 0.03533000 &   \\
1000 & 0.00013811 & 0.00001171 & 0.00037305 & 0.00552957 & 0.00079710 & 0.01267597 &   \\
2000 & 0.00007177 & 0.00000208 & 0.00018429 & 0.00347506 & 0.00065212 & 0.00823961 &   \\
5000 & 0.00001786 & 0.00000159 & 0.00004655 & 0.00145635 & 0.00022798 & 0.00369547 &   \\
\bottomrule
\end{tabular}
\caption{Estimates of $\lambda_0$ and $c_0$ obtained without edge correction, the underlying $F$ is given by (\ref{true_F_option1_inversion}).}\label{table_uncorrected_uniform}
\end{table}

\begin{table}[t!]
\centering
\begin{tabular}
    {
  r
  S[table-format=1.7,round-precision=3]
>{{{\lp}}} 
S[round-precision=2, table-format = 1.7,table-space-text-pre=\lp]
@{,\,} 
S[round-precision=2, table-format = 1.5,table-space-text-post=\rp]
<{{{\rp}}} 
    S[table-format=1.5, round-precision=3]
>{{{\lp}}} 
S[round-precision=2,table-format = 1.5,table-space-text-pre=\lp]
@{,\,} 
S[round-precision=2,table-format = 1.4,table-space-text-post=\rp]
<{{{\rp}}} 
@{}l@{}
}
\toprule
\multicolumn{1}{c}{} & \multicolumn{3}{c}{$|\bar{\lambda}_n - \lambda_0|$} & \multicolumn{3}{c}{$\Vert \bar{c}_n -c_0\Vert$} & \\
\cmidrule(r){1-1}\cmidrule(lr){2-4}\cmidrule(l){5-8}
  \multicolumn{1}{c}{$P_n$} &        {mean} &    \multicolumn{2}{c}{(2.5\%, 97.5\%)} & {mean} &    \multicolumn{2}{c}{(2.5\%, 97.5\%)} & \\
\cmidrule(r){1-1}\cmidrule(lr){2-4}\cmidrule(l){5-8}
500 & 0.00126079 & 0.00007698 & 0.00391072 & 0.04495845 & 0.00575331 & 0.13544627 &   \\
1000 & 0.00048939 & 0.00003040 & 0.00138623 & 0.02219952 & 0.00371886 & 0.06179314 &   \\
2000 & 0.00023915 & 0.00000748 & 0.00073627 & 0.01373968 & 0.00228794 & 0.04304565 &   \\
5000 & 0.00008694 & 0.00000292 & 0.00027106 & 0.00734665 & 0.00090572 & 0.02215045 &   \\
\bottomrule
\end{tabular}
\caption{Estimates of $\lambda_0$ and $c_0$ obtained with edge correction, the underlying $F$ is given by (\ref{true_F_option2_inversion}).}\label{table_edgecorrected_discrete}
\vspace{5pt}
\begin{tabular}
    {
  r
  S[table-format=1.7,round-precision=3]
>{{{\lp}}} 
S[round-precision=2, table-format = 1.7,table-space-text-pre=\lp]
@{,\,} 
S[round-precision=2, table-format = 1.6,table-space-text-post=\rp]
<{{{\rp}}} 
    S[table-format=1.5, round-precision=3]
>{{{\lp}}} 
S[round-precision=2,table-format = 1.5,table-space-text-pre=\lp]
@{,\,} 
S[round-precision=2,table-format = 1.4,table-space-text-post=\rp]
<{{{\rp}}} 
@{}l@{}
}
\toprule
\multicolumn{1}{c}{} & \multicolumn{3}{c}{$|\hat{\lambda}_n - \lambda_0|$} & \multicolumn{3}{c}{$\Vert \hat{c}_n -c_0\Vert$} & \\
\cmidrule(r){1-1}\cmidrule(lr){2-4}\cmidrule(l){5-8}
  \multicolumn{1}{c}{$P_n$} &        {mean} &    \multicolumn{2}{c}{(2.5\%, 97.5\%)} & {mean} &    \multicolumn{2}{c}{(2.5\%, 97.5\%)} & \\
\cmidrule(r){1-1}\cmidrule(lr){2-4}\cmidrule(l){5-8}
500 & 0.00093743 & 0.00004657 & 0.00289883 & 0.03386005 & 0.00421877 & 0.10824566 &   \\
1000 & 0.00041400 & 0.00001250 & 0.00100824 & 0.01940134 & 0.00191347 & 0.04515031 &   \\
2000 & 0.00015830 & 0.00001302 & 0.00046083 & 0.01036915 & 0.00111555 & 0.02791533 &   \\
5000 & 0.00004868 & 0.00000241 & 0.00015076 & 0.00439266 & 0.00054577 & 0.01104591 &   \\
\bottomrule
\end{tabular}
\caption{Estimates of $\lambda_0$ and $c_0$ obtained without edge correction, the underlying $F$ is given by (\ref{true_F_option2_inversion}).}\label{table_uncorrected_discrete}
\end{table}

For the simulations in this section we simulate Poisson-Laguerre tessellations in $\RR^2$ with the following choices for the underlying function $F$, with $z \geq 0$:
\begin{align}
    F_1(z) &= z\cdot\mathds{1}\{z < 1\} + \mathds{1}\{z \geq 1\} \label{true_F_option1_inversion}\\
    F_2(z) &= 0.01\cdot\mathds{1}\{z \geq 1\} + 0.04\cdot\mathds{1}\{z \geq 8\} + 0.95\cdot\mathds{1}\{z \geq 10\}.\label{true_F_option2_inversion}
\end{align}
For both choices of $F$ it is simple to simulate a corresponding Poisson process, because these Poisson processes can be recognized as independently marked homogeneous Poisson processes. Following \cite{vdjagt2025b} we define $\eta^0 = \{(x,h) \in \eta: x\in C((x,h),\eta)\}$. We write $P_n:=\EE(\eta^{0}(W_n \times (0,\infty))$, and we choose a square observation window $W_n$ such that $P_n = 1000$. In words, we choose a square $W_n$ with an area such that the expected number of observed points of $\eta^0$ in $W_n$ is equal to 1000. The reason for introducing this point process $\eta^0$ is because of its importance for the definition of $\hat{F}_n^0$. For each choice of $F$ and $P_n$, 100 Laguerre tessellations are simulated, yielding 100 realizations of $(\bar{\lambda}_n, \bar{c}_n)$ and $(\hat{\lambda}_n,\hat{c}_n)$. For the sake of convenience we take $\lambda_0 = 1$ and $c_0 = 0$. The results of these simulations are summarized in Tables \ref{table_edgecorrected_uniform}-\ref{table_uncorrected_discrete}. In these tables the average absolute errors are shown, as well as the $2.5\%$ and $97.5\%$ quantiles of these absolute errors.

Let us now discuss the contents of these tables. For one, it can be seen that the estimates of $\lambda_0$ and $c_0$ corresponding to $F_1$ are more accurate compared to estimates corresponding to $F_2$. A Poisson-Laguerre tessellation with $F_1$ as the underlying distribution function will have mostly cells of similar sizes. If $F_2$ is the underlying distribution function then a corresponding Poisson-Laguerre tessellation will have a few large cells, a larger number of medium sized cells and a lot of small cells. Perhaps estimation of $(\lambda_0,c_0)$ is more difficult when there is more variation in the cell sizes. As anticipated we can also see that the estimator $\smash{(\bar{\lambda}_n, \bar{c}_n)}$ yields larger mean absolute errors compared to $\smash{(\hat{\lambda}_n, \hat{c}_n)}$. Finally, it should be noted that the errors corresponding to estimates of $\lambda_0$ are far smaller compared to the errors corresponding to estimates of $c_0$. This is perhaps caused by the rather extreme behavior of the criterion function $T_n$ as $n \to \infty$ if $\lambda\neq\lambda_0$ (Theorem \ref{criterion_func_asymptotics}). Overall, both estimation procedures appear to yield rather close estimates of $\lambda_0$ and $c_0$. Indeed, the errors are sufficiently small that visualizing both the original configuration of generators as well as the generators obtained via any of the two inversion procedures ($\smash{(\bar{\lambda}_n, \bar{c}_n)}$ and $\smash{(\hat{\lambda}_n, \hat{c}_n)}$) we see that these points appear to overlap. This is shown in Figure \ref{figure_laguerre_inversion_example}. Here, the circle radii represent the weights corresponding to the generators. For visualization purposes the circle radii for the realization corresponding to $F_2$ were normalized (divided by 10). When looking very closely, one can see that the weights of the original generators are slightly different than the weights of the weighted generators obtained via the inversion procedure.

\begin{figure}[t!]
    \centering
    \makebox[\textwidth]{\makebox[\textwidth]{
    \begin{subfigure}[t]{0.5\textwidth}
        \centering
        \includegraphics[width=0.92\linewidth]{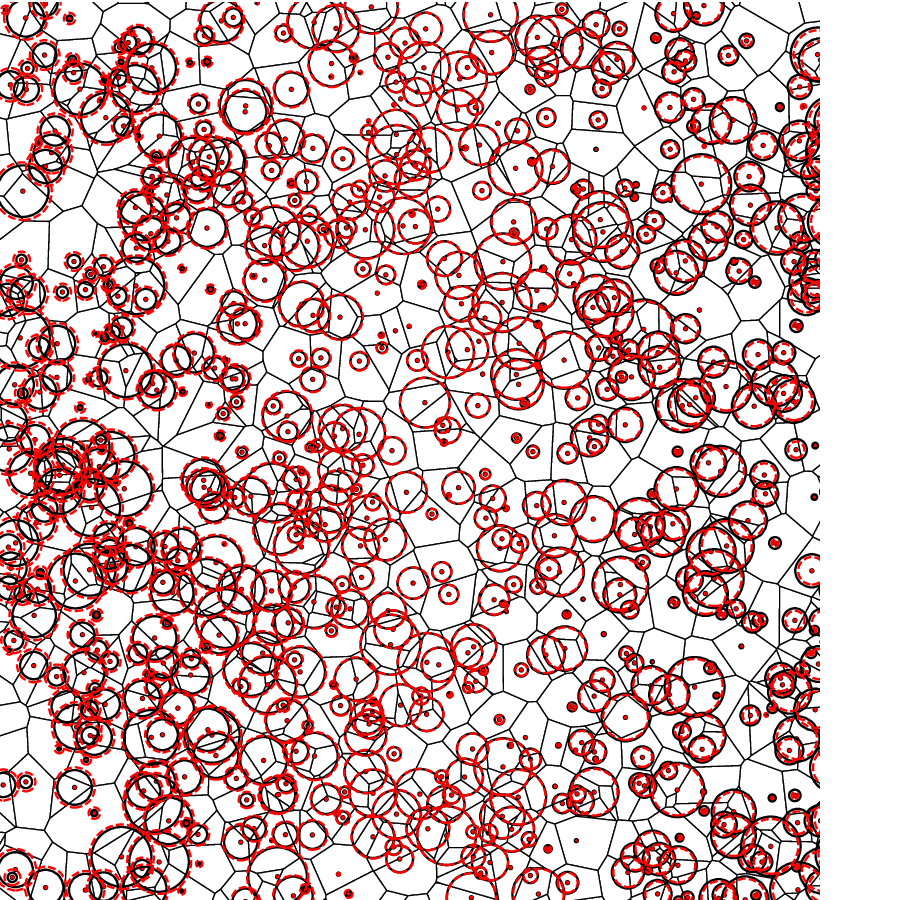}
    \end{subfigure}
    \begin{subfigure}[t]{0.5\textwidth}
        \centering
        \includegraphics[width=0.92\linewidth]{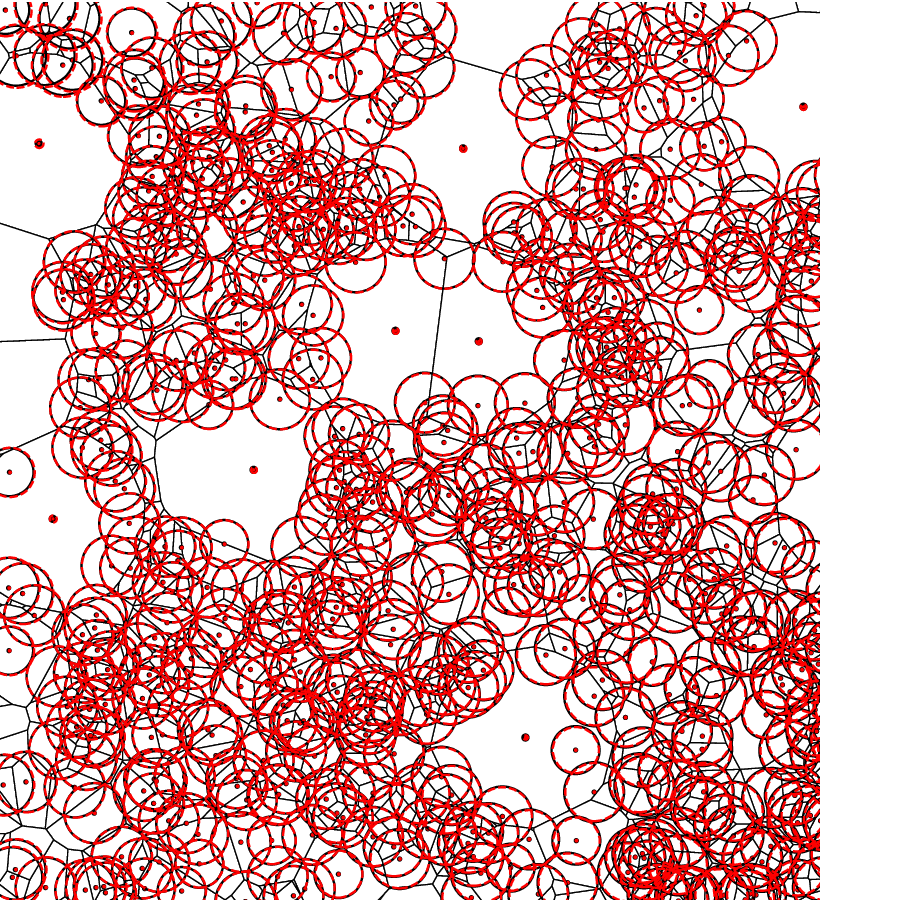}
    \end{subfigure}\hfill}}
    \caption{Realizations of Poisson-Laguerre tessellations with $F_1$ (Left) and $F_2$ (Right) as the underlying distribution function. Black points and black circles represent the original weighted generators ($\eta^*$), red points and red dashed circles represent weighted generators ($\smash{\hat{\eta}_n^*}$) obtained via the estimator $\smash{(\hat{\lambda}_n,\hat{c}_n)}$, see (\ref{inversion_point_process}).}
    \label{figure_laguerre_inversion_example}
\end{figure}

\subsection{Estimation of \texorpdfstring{$F$}{F}}
In this section we essentially investigate whether it is possible to estimate $F$, if the only available information is a region of a Poisson-Laguerre tessellation observed through a bounded observation window $W_n$. That is, the underlying weighted generators are considered to be unknown. This also means that cells at the boundary of the window are only partially observed. The simulations in this section may be seen as a continuation of the previous section, as the same choices for $F$ are considered, and the estimates of $F$ are computed based on weighted generators obtained via the inversion procedure defined via $\smash{(\bar{\lambda}_n, \bar{c}_n)}$.

Let $\hat{\eta}_n^*$ denote the configuration of weighted generators obtained via the inversion procedure defined via $(\bar{\lambda}_n, \bar{c}_n)$, recall (\ref{inversion_point_process}). Then, $\hat{\eta}_n^*$ may be seen as an approximation of $\eta^*\cap (W_n \times \RR)$. Write $\hat{\eta}_n^* = \{(\hat{x}_1,\hat{h}_1),\dots,(\hat{x}_m,\hat{h}_m)\}$, with $\hat{h}_1 \leq \hat{h}_2 \leq \dots \leq \hat{h}_m$. Then, analogously to the estimator $\hat{F}_n^0$ for $F$ as defined in \cite{vdjagt2025b}, we now define the estimator $\bar{F}_n^0$ for $F$ as follows. For $z \geq 0$ define:
\begin{equation}
    \bar{G}_n(z) := \frac{1}{\nu_d(W_n)}\sum_{(\hat{x},\hat{h}) \in \hat{\eta}_n^*} \mathds{1}_{W_n}(\hat{x})\mathds{1}_{(0,z]}(\hat{h}) \mathds{1}\{\hat{x} \in C((\hat{x},\hat{h}),\hat{\eta}_n^*) \}.
\end{equation}
This is essentially an adaptation of the estimator $\hat{G}_n$ for $G_F$ as defined in \cite{vdjagt2025b}. We define the estimator $\bar{F}_n^0$ as a piece-wise constant distribution function with jump locations at $\hat{h}_1,\dots,\hat{h}_m$. Hence, we only need to specify the values $\bar{F}_n^0(\hat{h}_i)$ for $i \in \{1,\dots,m\}$. Let $\hat{h}_0 < \hat{h}_1$ and set $\bar{F}_n^0(\hat{h}_0)=0$. For $i \in \{1,\dots,m\}$, the function $\bar{F}_n^0$ is recursively defined via:
\begin{align}%
\begin{split}\label{computational_formula_F_estimator_inversion}
    \bar{F}_n^0(\hat{h}_i) &= \bar{F}_n^0(\hat{h}_{i-1}) + \left(\bar{G}_n(\hat{h}_i) - \bar{G}_n(\hat{h}_{i-1}) \right)  \exp\left(\kappa_d \sum_{j=1}^{i-1} \left(\hat{h}_i -\hat{h}_j \right)^{\frac{d}{2}}\left(\bar{F}_n^0(\hat{h}_j) - \bar{F}_n^0(\hat{h}_{j-1}) \right) \right). 
\end{split}
\end{align}%
\begin{figure}[t!]
    \centering
    \includegraphics[width=0.5\linewidth]{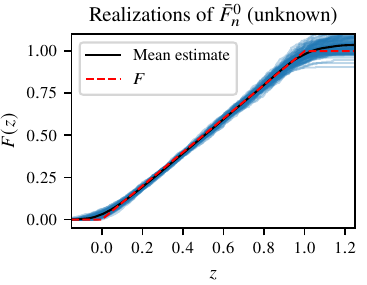}\includegraphics[width=0.5\linewidth]{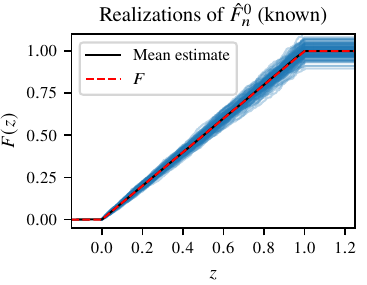}\\
     \includegraphics[width=0.5\linewidth]{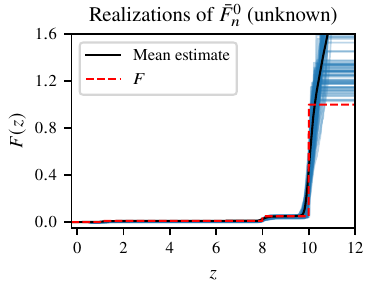}\includegraphics[width=0.5\linewidth]{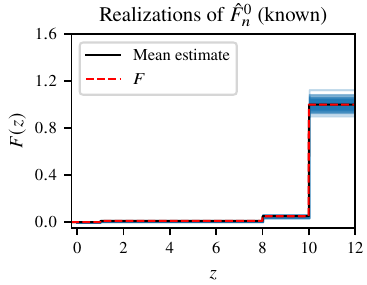}
    \caption{Realizations of the estimators $\bar{F}_n^0$ (left panel, weighted generators are unknown) and $\hat{F}_n^0$ (right panel, weighted generators are considered known). In the upper panel $F$ is given by (\ref{true_F_option1_inversion}), in the lower panel $F$ is given by (\ref{true_F_option2_inversion}). In these simulations, $P_n =1000$.}
    \label{figure_inversion_estimates}
\end{figure}%
As such, we now have an estimator $\bar{F}_n^0$ for $F$ which can be computed if the weighted generators of the (partially) observed cells are a priori unknown. We compare the obtained estimates to realizations of $\hat{F}_n^0$, which can be computed if the weighted generators of the (partially) observed cells are considered known. Then, we can see whether not knowing the generators has a noticeable effect on the resulting estimates. We should mention that for the obtained realizations of $\bar{F}_n^0$ it was needed to apply a shift. Recall that the distribution of a Poisson-Laguerre tessellation is not affected by shifts of the underlying distribution function. As such, if two estimates of $F$ are equal op to a shift, these estimates are considered equally accurate. Hence, to each realization of $\bar{F}_n^0$ a shift was applied which essentially minimizes the average distance to the true function $F$. The results of these simulations are shown in Figure \ref{figure_inversion_estimates}. 

First, consider the estimates corresponding to $F_1$. Visually, it appears as if $\bar{F}_n^0$ yields estimates which are smoothed versions of $\hat{F}_n^0$. Overall, both estimators appear to yield accurate estimates. The estimator $\bar{F}_n^0$ appears to be less accurate in the left and right tail of $F$, and it is very accurate in the middle of its support ($z$ close to $0.5$). Meanwhile, realizations of $\hat{F}_n^0$ are very accurate for $z$ close to 0 and slowly become more inaccurate as $z$ becomes larger. The estimates corresponding to $F_2$ paint a rather different picture. While both estimators appear to provide accurate estimates of $F$ for say $z < 9.5$, the estimates obtained via $\bar{F}_n^0$ become very inaccurate for larger values of $z$. Some of these realizations of $\bar{F}_n^0$ have an upper bound which is much too large to fit in the frame of the plot. We should mention here that it is not unexpected that estimation of $F(z)$ is more difficult for large values of $z$, as also discussed in \cite{vdjagt2025b}. Essentially, points with large weights are sampled much more rarely, since points with large weights often do not generate a cell. We also expect that for large values of $z$ changes in $F(z)$ only have a small effect on the distribution of the corresponding Poisson-Laguerre tessellation. However, at present we do not have a formal result to quantify this statement, it is not yet clear when an estimate of $F$ may be considered accurate.

\section{Concluding remarks}\label{section_discussion_inversion}
In this paper we derived an inversion procedure for Poisson-Laguerre tessellations. The main motivation for wanting to retrieve the weighted generators of the observed cells is for their use in statistical inference, as in \cite{vdjagt2025b}. We obtained various theoretical results and provided conditions for the inversion procedure to be consistent. In simulations we observed that the inversion procedure yields a very close approximation of the original weighted generators. Additionally, we have shown how these weighted generators may be used to estimate the distribution function $F$. It is however evident that if there is no prior knowledge of the original weighted generators, the resulting estimates of $F$ are less accurate compared to when estimates of $F$ can be computed based on the original weighted generators. Whether the estimators for $F$ based on weighted generators obtained via the inversion procedure are in general consistent is an important open problem. Finally, we should mention that our inversion procedure does not take into account edge effects. A simple way to deal with edge effects in practice is described in section \ref{section_simulations_inversion}, but perhaps more sophisticated approaches for dealing with edge effects can be incorporated into the inversion procedure in future research.

\section{Acknowledgements}
This paper is based on chapter 6 in the first author's PhD thesis \cite{vdjagt2026}.

\section{Additional proofs}\label{section_proofs_inversion}

\begin{proof}[Proof of Lemma \ref{pF_moments_lemma}]
Let $y \in \RR^d \setminus \{0\}$, by (\ref{lower_bound_fractional_integral}) we have:
\[ \int_0^{\Vert y \Vert^2 + h}\left(\Vert y \Vert^2 + h -t \right)^\frac{d}{2}\mathrm{d}F(t) \geq F(h) \Vert y \Vert^d.\]
Via this inequality we obtain:
\begin{align*}
    p_F(y) &= \int_0^\infty \exp\left(-\kappa_d \int_0^{\Vert y \Vert^2 +h}\left(\Vert y \Vert^2 +h -t \right)^{\frac{d}{2}}\mathrm{d}F(t) \right)\mathrm{d}F(h) \\
    &\leq \int_0^\infty \exp\left(-\kappa_d F(h) \Vert y \Vert^d \right)\mathrm{d}F(h)\\
   &= \int_{F(0)}^{F(\infty)} \exp\left(-\kappa_d u \Vert y \Vert^d \right)\mathrm{d}u \\
   &= \frac{1}{\kappa_d \Vert y \Vert^d}\left(\exp\left(-\kappa_d \cdot 0 \cdot \Vert y \Vert^d \right) - \exp\left(-\kappa_d \cdot F(\infty) \cdot \Vert y \Vert^d \right) \right) \\
   &\leq \frac{1}{\kappa_d \Vert y \Vert^d}.
\end{align*}
Here, we substituted $u = F(h)$. Let $q \in \NN$, then substituting $y = r\theta$, with $r \geq 0$ and $\theta \in \Sp^{d-1}$ we obtain:
{\allowdisplaybreaks
\begin{align*}
        \int_{\RR^d}\Vert y \Vert^qp_F(y)\mathrm{d}y &= \int_{\RR^d}\Vert y \Vert^q\int_0^\infty \exp\left(-\kappa_d \int_0^{\Vert y \Vert^2 +h}\left(\Vert y \Vert^2 +h -t \right)^{\frac{d}{2}}\mathrm{d}F(t) \right)\mathrm{d}F(h)\mathrm{d}y\\
        &= d \kappa_d  \int_0^\infty  \int_0^\infty \exp\left(-\kappa_d \int_0^{r^2 +h}\left(r^2 + h - t \right)^{\frac{d}{2}}\mathrm{d}F(t) \right)r^{q+d-1} \mathrm{d}r\mathrm{d}F(h).
    \end{align*}
}%
    This expression is very similar to an expression which appears in the proof of Theorem \ref{thm_integrated_second_moment}. Indeed, if we were to set $q = (p-1)d$, $p \in \NN$, then the proof of Theorem \ref{thm_integrated_second_moment} yields $\int_{\RR^d}\Vert y \Vert^qp_F(y)\mathrm{d}y < \infty$. Because Theorem \ref{thm_integrated_second_moment} holds for any $d, p \in \NN$ we may simply choose $p$ sufficiently large such that $q \leq (p-1)d$ and then we obtain the result via H\"older's inequality:
    \[\int_{\RR^d}\Vert y \Vert^qp_F(y)\mathrm{d}y \leq \left(\int_{\RR^d}\Vert y \Vert^{(p-1)d}p_F(y)\mathrm{d}y\right)^\frac{(p-1)d}{q} < \infty.\]
\end{proof}

\begin{proof}[Proof of Lemma  \ref{cell_in_ball_lemma}]
    Assume $\varphi \cap B_{x,k} \neq \emptyset$ for all $k \in \{1,\dots,J\}$. Let $y \in C((x,h),\varphi) - x$, then also $y \in \mathcal{C}_j$ for some $j \in \{1,\dots,J\}$. Choose $x_j \in \varphi \cap B_{x,j}$. Then, $2\sqrt{z} < \Vert x_j - x \Vert < R$ and $x_j - x \in \mathcal{C}_j$. Because $y \in C((x,h),\varphi) - x$ and $x_j \in \varphi$ we have:
    \begin{align*}
        \Vert y \Vert^2 +h \leq \Vert x_j - x - y \Vert^2 + h_j = \Vert x_j -x\Vert^2 + \Vert y \Vert^2 -2 \langle x_j -x,y\rangle + h_j
    \end{align*}
    Using $2\sqrt{z} < \Vert x_j - x \Vert$ we now obtain:
    \begin{equation}
        2\langle x_j - x,y\rangle \leq  \Vert x_j -x\Vert^2 + h_j - h \leq \Vert x_j -x\Vert^2 + z < \frac{5}{4}\Vert x_j - x \Vert^2.\label{circumscribed_ball_proof_eq}
    \end{equation}
    Because $y \in \mathcal{C}_j$, $x_j -x \in \mathcal{C}_j$ we have by the definition of $\mathcal{C}_j$:
    \[\langle x_j - x,y\rangle \leq \frac{3}{4}\Vert x_j - x\Vert \cdot \Vert y \Vert \iff \Vert y \Vert \leq \frac{4}{3}\frac{\langle x_j -x,y\rangle}{\Vert x_j - x \Vert}\]
    Combining this with (\ref{circumscribed_ball_proof_eq}) we obtain the desired result:
    \[\Vert y \Vert \leq \frac{4}{3}\frac{\langle x_j -x,y\rangle}{\Vert x_j - x \Vert} < \frac{5}{6} \frac{\Vert x_j - x \Vert^2}{\Vert x_j - x \Vert} < \Vert x_j - x \Vert < R.\]
    Indeed, $C((x,h),\varphi) - x \subset B(0,R) \iff C((x,h),\varphi) \subset B(x,R)$. 
\end{proof}

\begin{proof}[Proof of Theorem \ref{criterion_func_asymptotics}]
    First, consider the case $\lambda = \lambda_0$. Equation (\ref{def_criterion_function}) may be written as:
    \begin{align*}
        T_n(\lambda_0,c) &= \frac{1}{\nu_d(W_n)}\sum_{(x,h) \in \eta} \int_{C((x,h),\eta)} \Vert x - c_0 +c - y \Vert^2 \mathrm{d}y \\
        &= \frac{1}{\nu_d(W_n)}\sum_{(x,h) \in \eta} \int_{C((x,h),\eta)-x} \Vert c - c_0 - y \Vert^2 \mathrm{d}y \\
        &= \frac{1}{\nu_d(W_n)}\sum_{(x,h) \in \eta} \int_{C((0,h),S_x\eta)} \Vert c - c_0 - y \Vert^2 \mathrm{d}y
    \end{align*}
    Applying the spatial ergodic theorem (see Proposition 13.4.I. in \cite{Daley2008}), $T_n(\lambda_0,c)$ converges almost surely to its expected value (see Lemma \ref{criterion_function_expected_value_lemma}). Now, consider the case $\lambda \neq \lambda_0$. Let $(\Omega,\mathcal{A},\PP)$ be a probability space supporting the Poisson process $\eta$. In order to show that $\lim_{n \to \infty} T_n(\lambda,c)=\infty$ almost surely, it is sufficient to show that for all $M > 0$ there exists a $\Omega_M \in \mathcal{A}$ with $\PP(\Omega_M)=1$ such that for all $\omega \in \Omega_M$: $\liminf_{n \to \infty}T_n(\lambda,c;\omega) > M$. Indeed, setting $\Omega_0 = \cap_{M \in \NN}\Omega_M$, we have $\PP(\Omega_0) = 1$ and $\lim_{n \to \infty}T_n(\lambda,c;\omega) = \infty$ for all $\omega \in \Omega_0$. Let $M > 0$. We need some additional notation before we choose $\Omega_M$. Pick $z > 0$ such that $F(z) > 0$. Take $R >0$ large enough such that $R > 2\sqrt{z}$. Take $r > 0$ small enough such that $r + \sqrt{r^2 + z} < 2\sqrt{z}$. One may for instance take $r = \smash{\sqrt{z/3}}$. For any $(x,h) \in \RR^d \times (0,\infty)$ let the set $D_{x,h}$ be as in Lemma \ref{ball_in_cell_lemma} and the sets $B_{x,j}$, $j \in \{1,\dots,J\}$ as in Lemma \ref{cell_in_ball_lemma}. For $(x,h) \in \RR^d \times(0,\infty)$ define the following event:
    \[O_{x,h} = \{\eta(D_{x,h}\setminus \{(x,h)\})=0, \eta(B_{x,j}) > 0 \ \ \forall j \in \{1,\dots,J\}\}.\]
    Note that $D_{x,h} \cap B_{x,j} = \emptyset$ for all $j \in \{1,\dots,J\}$. Additionally, the $B_{x,j}$'s have disjoint interiors and therefore we have:
    \[\PP(O_{x,h}) = \PP(\eta(D_{x,h}\setminus \{(x,h)\})=0)\prod_{j=1}^J \PP\left(\eta(B_{x,j}) > 0  \right).\]
    We will now argue that $\PP(O_{x,h}) = \PP(O_{0,h}) \geq \alpha > 0$, where $\alpha$ is a constant which does not depend on $h$. Note that proving this implies the following:
    \begin{equation}
        \int_0^z \PP(O_{0,h})\mathrm{d}F(h) \geq \alpha F(z) > 0.
    \end{equation}
    Since $(\nu_d \times \mathbb{F})(\{(x,h)\}) = 0$ we have $\PP(\eta(D_{x,h}\setminus \{(x,h)\})=0) = \PP(\eta(D_{x,h}) =0)$, which may be computed as:
    {\allowdisplaybreaks
    \begin{align*}
        \PP(\eta(D_{x,h}) =0) &= \exp\left(- \int_{\RR^d}\int_0^{r^2+h}\mathds{1}\{\Vert x -x'\Vert < r + \sqrt{r^2 + h -h'}\}\mathrm{d}F(h')\mathrm{d}x' \right)\\
        &= \exp\left(-\kappa_d \int_0^{r^2+h}\left(r + \sqrt{r^2 + h -h'} \right)^d\mathrm{d}F(h') \right)\\
        &\geq  \exp\left(-\kappa_d \int_0^{r^2+z}\left(z + \sqrt{r^2 + z} \right)^d\mathrm{d}F(h') \right) \\
        &= \exp\left(-\kappa_d F(r^2 + z)\left(z + \sqrt{r^2 + z} \right)^d \right) > 0.
    \end{align*}
    }
    Let $j \in \{1,\dots, J\}$, then:
    \begin{align*}
        \PP\left(\eta(B_{x,j}) > 0  \right) &= 1- \PP\left(\eta(B_{x,j}) = 0  \right)\\
        &= 1 - \exp\left(- F(z)\nu_d\left(\left(B(x,R)\setminus\bar{B}(x,2\sqrt{z})\right)  \cap (\mathcal{C}_j + x)\right)\right) \\
        &= 1 - \exp\left(- F(z)\nu_d\left(\left(B(0,R)\setminus\bar{B}(0,2\sqrt{z})\right)  \cap \mathcal{C}_j\right)\right) > 0.
    \end{align*}
    So indeed, there exists some $\alpha > 0$ such that for all $h \in [0,z]$ we have $\PP(O_{x,h}) = \PP(O_{0,h}) \geq \alpha > 0$. We now use the event $O_{x,h}$ to obtain a lower bound for $T_n(\lambda,c)$:
    \begin{align*}
        T_n(\lambda,c) &\geq\frac{1}{\nu_d(W_n)}\sum_{(x,h) \in \eta} \mathds{1}_{W_n}(x)\mathds{1}_{(0,z]}(h)v_{x,h}\left\Vert \left(\frac{\lambda}{\lambda_0}-1\right)x -\frac{\lambda c_0}{\lambda_0} + c - (c_{x,h}-x) \right\Vert^2 \\
        &\geq \frac{1}{\nu_d(W_n)}\sum_{(x,h) \in \eta} \mathds{1}_{W_n}(x)\mathds{1}_{(0,z]}(h)v_{x,h}\left\Vert \left(\frac{\lambda}{\lambda_0}-1\right)x -\frac{\lambda c_0}{\lambda_0} + c - (c_{x,h}-x) \right\Vert^2 \mathds{1}\{O_{x,h}\}\\
        &\geq \frac{\kappa_d r^d}{\nu_d(W_n)}\sum_{(x,h) \in \eta} \mathds{1}_{W_n}(x)\mathds{1}_{(0,z]}(h)\left\Vert \left(\frac{\lambda}{\lambda_0}-1\right)x -\frac{\lambda c_0}{\lambda_0} + c - (c_{x,h}-x) \right\Vert^2 \mathds{1}\{O_{x,h}\}.
    \end{align*}
    The final inequality follows from the fact that if the event $O_{x,h}$ occurs then $B(x,r) \subset C((x,h),\eta)$, hence $v_{x,h} \geq \kappa_d r^d$. Let $N > 0$ be large enough such that:
    \[\kappa_d r^d\left(\frac{N - \Vert c_0 \Vert}{\lambda_0}\left|\lambda-\lambda_0 \right| - \Vert c - c_0 \Vert - R \right)^2 \int_0^z \PP(O_{0,h})\mathrm{d}F(h) > M.\]
    Choose $\Omega_M \in \mathcal{A}$ with $\PP(\Omega_M) = 1$ such that for all $\omega \in \Omega_M$ the following two properties hold:
 \[\lim_{n \to \infty}\frac{1}{\nu_d(W_n)}\sum_{(x,h) \in \eta(\blank;\omega)} \mathds{1}_{W_n}(x)\mathds{1}_{(0,z]}(h)\mathds{1}\{O_{x,h}(\omega)\} = \int_0^z \PP(O_{0,h})\mathrm{d}F(h)\]
\[\eta(\bar{B}(0,N) \times (0,z];\omega) < \infty\]
Such a set exists by the spatial ergodic theorem (Proposition 13.4.I. in \cite{Daley2008}) and by the fact that a Poisson random variable with a finite rate parameter is almost surely finite. Suppose $(x,h) \in \eta$, $\Vert x \Vert > N$ and the event $O_{x,h}$ occurs such that $C((x,h),\eta) \subset B(x,R)$ and therefore $\Vert c_{x,h}-x\Vert < R$. Then, we claim that the following holds:
    \begin{equation}
        \left\Vert \left(\frac{\lambda}{\lambda_0}-1\right)x -\frac{\lambda c_0}{\lambda_0} + c - (c_{x,h}-x) \right\Vert \geq \frac{N-\Vert c_0 \Vert}{\lambda_0}|\lambda - \lambda_0| - \Vert c - c_0\Vert - R.\label{distance_term_lower_bound}
    \end{equation}
    We will now verify this claim.
    \begin{align}
        \left\Vert \left(\frac{\lambda}{\lambda_0}-1\right)x -\frac{\lambda c_0}{\lambda_0} + c - (c_{x,h}-x) \right\Vert &= \frac{1}{\lambda_0} \left|\lambda-\lambda_0 \right|\cdot \left\Vert x - \left(\frac{\lambda c_0 - \lambda_0 c + \lambda_0(c_{x,h}-x)}{\lambda - \lambda_0} \right) \right\Vert. \label{distance_to_generator_term}
    \end{align}
    Consider the term in brackets in the RHS of (\ref{distance_to_generator_term}), we obtain the following upper bound on its norm:
    \begin{align*}
        \left\Vert\frac{\lambda c_0 - \lambda_0 c + \lambda_0(c_{x,h}-x)}{\lambda - \lambda_0} \right\Vert &\leq \frac{1}{|\lambda - \lambda_0|}\left(\left\Vert \lambda c_0 - \lambda_0 c\right\Vert + \lambda_0 \left\Vert c_{x,h}-x \right\Vert \right) \\
        &\leq  \frac{1}{|\lambda - \lambda_0|}\left(\left\Vert \lambda c_0 - \lambda_0 c_0\right\Vert + \left\Vert \lambda_0c_0 - \lambda_0c \right\Vert + \lambda_0 R \right) \\
        &\leq \frac{1}{|\lambda - \lambda_0|}\left(|\lambda - \lambda_0|\Vert c_0 \Vert + \lambda_0 \Vert c - c_0 \Vert + \lambda_0 R\right) \\
        &\leq \Vert c_0 \Vert + \frac{\lambda_0}{|\lambda - \lambda_0|}\left(\Vert c - c_0 \Vert + R\right).
    \end{align*}
    Note that for $x,y \in \RR^d$ with $\Vert x \Vert > N > K > \Vert y \Vert$ we have by the reverse triangle inequality: $\Vert x - y\Vert \geq | \Vert x\Vert - \Vert y \Vert | = \Vert x \Vert - \Vert y \Vert > N - K$. Applying this to the RHS of (\ref{distance_to_generator_term}) we obtain the inequality in (\ref{distance_term_lower_bound}):
    \begin{align*}
        \left\Vert \left(\frac{\lambda}{\lambda_0}-1\right)x -\frac{\lambda c_0}{\lambda_0} + c - (c_{x,h}-x) \right\Vert 
        &\geq\frac{1}{\lambda_0}\left|\lambda-\lambda_0 \right|\left(N - \left(\Vert c_0 \Vert + \frac{\lambda_0}{|\lambda - \lambda_0|}\left(\Vert c - c_0 \Vert + R\right)\right) \right) \\
        &= \frac{N - \Vert c_0 \Vert}{\lambda_0}\left|\lambda-\lambda_0 \right| - \Vert c - c_0 \Vert - R.
    \end{align*}
    As a consequence, we may write (where $\bar{B}(0,N)^c = \RR^d \setminus \bar{B}(0,N)$):
    {\allowdisplaybreaks
    \begin{align}
        T_n(\lambda,c)  &\geq \frac{\kappa_d r^d}{\nu_d(W_n)}\sum_{(x,h) \in \eta} \mathds{1}_{W_n \cap \bar{B}(0,N)^c}(x)\mathds{1}_{(0,z]}(h)\left\Vert \left(\frac{\lambda}{\lambda_0}-1\right)x -\frac{\lambda c_0}{\lambda_0} + c - (c_{x,h}-x) \right\Vert^2 \mathds{1}\{O_{x,h}\}\nonumber \\
        &\phantom{\geq} + \frac{\kappa_d r^d}{\nu_d(W_n)}\sum_{(x,h) \in \eta} \mathds{1}_{W_n \cap \bar{B}(0,N)}(x)\mathds{1}_{(0,z]}(h)\left\Vert \left(\frac{\lambda}{\lambda_0}-1\right)x -\frac{\lambda c_0}{\lambda_0} + c - (c_{x,h}-x) \right\Vert^2 \mathds{1}\{O_{x,h}\}\nonumber \\
        \begin{split}\label{Tn_limit_two_terms}
            &\geq \frac{\kappa_d r^d}{\nu_d(W_n)}\sum_{(x,h) \in \eta} \mathds{1}_{W_n \cap \bar{B}(0,N)^c}(x)\mathds{1}_{(0,z]}(h)\left(\frac{N - \Vert c_0 \Vert}{\lambda_0}\left|\lambda-\lambda_0 \right| - \Vert c - c_0 \Vert - R \right)^2 \mathds{1}\{O_{x,h}\} \\
        &\phantom{\geq} + \frac{\kappa_d r^d}{\nu_d(W_n)}\sum_{(x,h) \in \eta} \mathds{1}_{W_n \cap \bar{B}(0,N)}(x)\mathds{1}_{(0,z]}(h)\left\Vert \left(\frac{\lambda}{\lambda_0}-1\right)x -\frac{\lambda c_0}{\lambda_0} + c - (c_{x,h}-x) \right\Vert^2 \mathds{1}\{O_{x,h}\}
        \end{split}
    \end{align}
    }
    We now consider the two terms in (\ref{Tn_limit_two_terms}) separately. The first term may be written as:
    {\allowdisplaybreaks
    \begin{align}
        &\frac{\kappa_d r^d}{\nu_d(W_n)}\left(\frac{N - \Vert c_0 \Vert}{\lambda_0}\left|\lambda-\lambda_0 \right| - \Vert c - c_0 \Vert - R \right)^2\sum_{(x,h) \in \eta} \mathds{1}_{W_n}(x)\mathds{1}_{(0,z]}(h)\mathds{1}\{O_{x,h}\} + \label{Tn_term1_part1} \\
        &- \frac{\kappa_d r^d}{\nu_d(W_n)}\left(\frac{N - \Vert c_0 \Vert}{\lambda_0}\left|\lambda-\lambda_0 \right| - \Vert c - c_0 \Vert - R \right)^2\sum_{(x,h) \in \eta} \mathds{1}_{W_n\cap \bar{B}(0,N)}(x)\mathds{1}_{(0,z]}(h) \mathds{1}\{O_{x,h}\}\label{Tn_term1_part2}.
    \end{align}
    }%
    By the choice of $\Omega_M$, the expression in (\ref{Tn_term1_part1}) converges for all $\omega \in \Omega_M$ to:
    \[\kappa_d r^d\left(\frac{N - \Vert c_0 \Vert}{\lambda_0}\left|\lambda-\lambda_0 \right| - \Vert c - c_0 \Vert - R \right)^2 \int_0^z \PP(O_{0,h})\mathrm{d}F(h).\]
    The expression in (\ref{Tn_term1_part2}) converges for all $\omega \in \Omega_M$ to $0$, since we have:
    \begin{align*}
        0 &\leq \frac{\kappa_d r^d}{\nu_d(W_n)}\left(\frac{N - \Vert c_0 \Vert}{\lambda_0}\left|\lambda-\lambda_0 \right| - \Vert c - c_0 \Vert - R \right)^2 \sum_{(x,h) \in \eta} \mathds{1}_{W_n \cap \bar{B}(0,N)}(x)\mathds{1}_{(0,z]}(h) \mathds{1}\{O_{x,h}\} \\
        &\leq \kappa_d r^d \left(\frac{N - \Vert c_0 \Vert}{\lambda_0}\left|\lambda-\lambda_0 \right| - \Vert c - c_0 \Vert - R \right)^2 \frac{\eta((W_n \cap \bar{B}(0,N))\times(0,z])}{\nu_d(W_n)}.
    \end{align*}
    Indeed, for sufficiently large $n$ we have $W_n \cap \bar{B}(0,N) = \bar{B}(0,N)$ and since $\eta(\bar{B}(0,N) \times (0,z];\omega) < \infty$ for all $\omega \in \Omega_M$, the fact that $\nu_d(W_n) \to \infty$ as $n \to \infty$ yields the convergence to zero. A similar argument may be used to show that the second term of (\ref{Tn_limit_two_terms}) converges to zero, this follows from the following inequalities:
    \begin{align*}
        0 &\leq \frac{\kappa_d r^d}{\nu_d(W_n)}\sum_{(x,h) \in \eta}\mathds{1}_{W_n \cap \bar{B}(0,N)}(x)\mathds{1}_{(0,z]}(h)\left\Vert \left(\frac{\lambda}{\lambda_0}-1\right)x -\frac{\lambda c_0}{\lambda_0} + c - (c_{x,h}-x) \right\Vert^2 \mathds{1}\{O_{x,h}\} \\
        &\leq \kappa_d r^d \left(\left|\frac{\lambda}{\lambda_0} - 1 \right|N + \left\Vert c - \frac{\lambda c_0}{\lambda_0} \right\Vert + R \right)^2 \frac{\eta((W_n \cap \bar{B}(0,N))\times(0,z])}{\nu_d(W_n)}.
    \end{align*}
    Hence, combining all results, we have for all $\omega \in \Omega_M$:
    \[\liminf_{n \to \infty} T_n(\lambda,c;\omega) \geq \kappa_d r^d\left(\frac{N - \Vert c_0 \Vert}{\lambda_0}\left|\lambda-\lambda_0 \right| - \Vert c - c_0 \Vert - R \right)^2 \int_0^z \PP(O_{0,h})\mathrm{d}F(h) > M.\]
\end{proof}

\begin{proof}[Proof of Proposition \ref{prop_second_moment_bound}]
    We start with proving the first inequality in the proposition for $p \geq 2$. Note that: $\nu_d(C((x,h),\eta)) = \nu_d(C((x,h),\eta)-x)$ by the translation invariance of Lebesgue measure. For $y \in \RR^d$ recall that $y \in C((x,h),\eta)-x \iff \eta(A_{x,h,y}) = 0$ with $A_{x,h,y}$ as in (\ref{set_A_xhy}). By (\ref{Robbins_thm}) we now obtain:
    \begin{equation}
        \EE\left(\nu_d(C((x,h),\eta))^p\right) = \int_{\RR^d}\dots\int_{\RR^d}\PP\left(\eta(A_{x,h,y_1}) = 0,\dots, \eta(A_{x,h,y_p}) = 0\right)\mathrm{d}y_1 \dots\mathrm{d}y_p.\label{robbin_thm_application}
    \end{equation}
    Instead of integrating over $(\RR^d)^p$ we may also integrate over the union of all sets of the form:
    \begin{equation*}
        \mathcal{C}_I =\left\{(y_1,\dots,y_p) \in (\RR^d)^p: \Vert y_{i_1} \Vert \leq \Vert y_{i_2}\Vert \leq \dots \leq \Vert y_{i_p} \Vert  \right\},
    \end{equation*}
    where $I =(i_1,i_2,\dots,i_p)$ is a permutation of $(1,2,\dots,p)$. Note that the integrand in (\ref{robbin_thm_application}) is symmetric in $y_1,\dots,y_p$. Hence, when integrating over any set $\mathcal{C}_I$ the result is the same for every choice of $I$. Because there are $p!$ permutations of $(1,2,\dots,p)$, we may write:
    \begin{align}
        \EE\left(\nu_d(C((x,h),\eta))^p\right)  &=p!\idotsint_\mathcal{C}\PP\left(\eta(A_{x,h,y_1}) = 0,\dots, \eta(A_{x,h,y_p}) = 0\right)\mathrm{d}y_1 \dots\mathrm{d}y_p \nonumber \\
        &\leq p!\idotsint_\mathcal{C}\PP\left(\eta(A_{x,h,y_p}) = 0\right)\mathrm{d}y_1 \dots\mathrm{d}y_p \nonumber \\ 
        &= p!\idotsint_\mathcal{C}\exp\left(-\kappa_d\int_0^{\Vert y_p\Vert^2 + h}\left(\Vert y_p\Vert^2 + h -t \right)^\frac{d}{2}\mathrm{d}F(t)\right)\mathrm{d}y_1 \dots\mathrm{d}y_p. \label{integral_bound}
    \end{align}
    The final equality follows from (\ref{point_inclusion_probability}) and $\mathcal{C}$ is given by:
     \[\mathcal{C} = \left\{(y_1,\dots,y_p) \in (\RR^d)^p: \Vert y_{1} \Vert \leq \Vert y_{2}\Vert \leq \dots \leq \Vert y_{p} \Vert  \right\}.\]
    Next, we substitute $y_i = r_i\theta_i$ with $r_i \geq 0$ and $\theta_i \in \Sp^{d-1}$ for $i \in \{1,\dots,p\}$. Then (\ref{integral_bound}) is given by:
    \begin{align}
         & p!(d\kappa_d)^p \int_0^\infty \int_0^{r_{p}} \dots \int_0^{r_3} \int_0^{r_2} \varphi(r_p)r_1^{d-1}\cdots r_p^{d-1}\mathrm{d}r_1\dots \mathrm{d}r_p \nonumber \\
        &= p!(d\kappa_d)^p \int_0^\infty \varphi(r_p) \left( \int_0^{r_{p}} \dots \int_0^{r_3} \int_0^{r_2}r_1^{d-1}\cdots r_p^{d-1}\mathrm{d}r_1\dots \mathrm{d}r_{p-1}\right)\mathrm{d}r_p, \label{intermediate_integral}
    \end{align}
    with: 
    \[\varphi(r_p) = \exp\left(-\kappa_d\int_0^{r_p^2 + h}\left(r_p^2 + h -t \right)^\frac{d}{2}\mathrm{d}F(t)\right).\]
    The integral in brackets in (\ref{intermediate_integral}) may be computed as:    \[\int_0^{r_{p}} \dots \int_0^{r_3} \int_0^{r_2}r_1^{d-1}\cdots r_p^{d-1}\mathrm{d}r_1\dots \mathrm{d}r_{p-1} = \frac{1}{(p-1)! d^{p-1}}r_p^{dp-1}.\]
    Plugging this back into (\ref{intermediate_integral}) we obtain the first inequality of the proposition:
    \begin{align}
        \EE\left(\nu_d(C((x,h),\eta))^p\right) &\leq pd\kappa_d^p\int_0^\infty \exp\left(-\kappa_d\int_0^{r_p^2 + h}\left(r_p^2 + h -t \right)^\frac{d}{2}\mathrm{d}F(t)\right)r_p^{pd-1}\mathrm{d}r_p. \label{prop_p_moment_result}
    \end{align}
    In the case $p=1$ (\ref{Robbins_thm}) may be used to show that (\ref{prop_p_moment_result}) becomes an equality as none of the techniques used to obtain an upper bound for the case $p \geq 2$ need to be used. Consider the second inequality for general $p \in \NN$. Via integration by parts, we obtain:
    \begin{align}
        \int_0^{r_p^2 + h}\left(r_p^2 + h -t \right)^\frac{d}{2}\mathrm{d}F(t) &= \frac{d}{2}\int_0^{r_p^2 + h}\left(r_p^2 + h -t \right)^{\frac{d}{2}-1}F(t)\mathrm{d}t \nonumber \\
        &\geq \frac{d}{2}\int_h^{r_p^2 + h}\left(r_p^2 + h -t \right)^{\frac{d}{2}-1}F(t)\mathrm{d}t \nonumber\\
        &\geq F(h)\frac{d}{2}\int_h^{r_p^2 + h}\left(r_p^2 + h -t \right)^{\frac{d}{2}-1}\mathrm{d}t\nonumber \\
        &= F(h)r_p^d. \label{lower_bound_fractional_integral}
    \end{align}
    We now plug this lower bound into (\ref{prop_p_moment_result}) and assume $F(h) > 0$:
    \begin{align*}
        \EE\left(\nu_d(C((x,h),\eta))^p\right) &\leq pd\kappa_d^p\int_0^\infty \exp\left(-\kappa_dF(h)r_p^d\right)r_p^{pd-1}\mathrm{d}r_p\\
        &= \frac{p}{F(h)^p}\int_0^\infty e^{-u}u^{p-1}\mathrm{d}u\\ &= \frac{p!}{F(h)^p}.
    \end{align*}
    Here, we substituted $u = \kappa_d F(h)r_p^d$. The resulting integral may be recognized as the Gamma function evaluated in $p$.
\end{proof}

\begin{proof}[Proof of Lemma \ref{lemma_individual_second_moment_finite}]
    We derive the result by showing that the first upper bound in the statement of Proposition \ref{prop_second_moment_bound} is finite. Choose $z > 1$ sufficiently large such that $z - \sqrt{z} > 1$ and $F(z) > 0$.
    \begin{align}
        \EE&\left(\nu_d(C((x,h),\eta))^p\right) \nonumber  \\
        &\leq pd\kappa_d^p\int_0^zr^{pd-1}\mathrm{d}r + pd\kappa_d^p\int_z^\infty \exp\left(-\kappa_d\int_0^{r^2 + h}\left(r^2 + h -t \right)^\frac{d}{2}\mathrm{d}F(t)\right)r^{pd-1}\mathrm{d}r \nonumber \\
        &= \kappa_d^pz^{pd} + pd\kappa_d^p\int_z^\infty \exp\left(-\kappa_d\int_0^{r^2 + h}\left(r^2 + h -t \right)^\frac{d}{2}\mathrm{d}F(t)\right)r^{pd-1}\mathrm{d}r. \label{second_moment_finite_first_bound}
    \end{align}
     Via integration by parts, we obtain:
     {\allowdisplaybreaks
     \begin{align*}
        \int_0^{r^2 + h}\left(r^2 + h -t \right)^\frac{d}{2}\mathrm{d}F(t) &= \frac{d}{2}\int_0^{r^2 + h}\left(r^2 + h -t \right)^{\frac{d}{2}-1}F(t)\mathrm{d}t \\
        &\geq \frac{d}{2}\int_z^{r^2 + h}\left(r^2 + h -t \right)^{\frac{d}{2}-1}F(z)\mathrm{d}t \\
        &\geq F(z)\frac{d}{2}\int_z^{r^2 + h}\left(r^2 + h -t \right)^{\frac{d}{2}-1}\mathrm{d}t \\
        &= F(z)(r^2 +h - z)^\frac{d}{2}\\
        &\geq F(z)(r^2 - z)^\frac{d}{2} \\
        &= F(z)(r-\sqrt{z})^\frac{d}{2}(r+ \sqrt{z})^\frac{d}{2}\\
        &\geq F(z)r^\frac{d}{2}.
    \end{align*}
     }%
    Plugging this back into (\ref{second_moment_finite_first_bound}) and substituting $u = \kappa_d r^\frac{d}{2}F(z)$ as in the proof of Proposition \ref{prop_second_moment_bound} yields:
    \begin{align*}
        \EE\left(\nu_d(C((x,h),\eta))^p\right) &\leq \kappa_d^p z^{pd} + pd\kappa_d^p\int_z^\infty \exp\left(-\kappa_d F(z)r^{\frac{d}{2}}\right)r^{pd-1}\mathrm{d}r \\
        &\leq \kappa_d^pz^{pd} + \kappa_d^p z^{pd} + pd\kappa_d^p\int_0^\infty \exp\left(-\kappa_d F(z)r^{\frac{d}{2}}\right)r^{pd-1}\mathrm{d}r \\
        &= \kappa_d^pz^{pd} + \frac{p!}{F(z)^p} =:\alpha_p.
    \end{align*}
\end{proof}

\begin{proof}[Proof of Theorem \ref{thm_integrated_second_moment}]
    For $p=1$ the result follows from (\ref{mean_cell_volume1}). For the remainder of the proof consider $p \geq 2$. Choose $z > 0$ large enough such that $F(z) > 0$. Let $0 < \alpha_p < \infty$ be as Lemma \ref{lemma_individual_second_moment_finite}. Using the second statement of Proposition \ref{prop_second_moment_bound} we may write:
    {\allowdisplaybreaks
    \begin{align}
        \int_0^\infty \EE\left(\nu_d(C((x,h),\eta))^p\right) \mathrm{d}F(h) &\leq \int_0^z\alpha_p\mathrm{d}F(h) + \int_z^\infty\EE\left(\nu_d(C((x,h),\eta))^p\right)\mathrm{d}F(h) \nonumber \\
        &\leq \alpha_p F(z) + \int_z^\infty \frac{p!}{F(h)^p}\mathrm{d}F(h)\nonumber\\
        &= \alpha_pF(z) + \int_{F(z)}^{F(\infty)}\frac{p!}{u^p}\mathrm{d}u \label{integrated_second_moment_sub}\\
        &= \alpha_pF(z) + \frac{p!}{p-1}\left(\frac{1}{F(z)^{p-1}} - \frac{1}{F(\infty)^{p-1}} \right) \nonumber\\
        &\leq \alpha_pF(z) + \frac{p!}{(p-1)F(z)^{p-1}} =:\beta_p. \nonumber
    \end{align}}
    In (\ref{integrated_second_moment_sub}) we substituted $u = F(h)$.
\end{proof}

\begin{proof}[Proof of Lemma \ref{lemma_bound_product_cell_volumes}]
    Let $h_1,h_2 > 0$ and define the following two half spaces:
    \begin{align*}
        H_1 &= \left\{y \in \RR^d: \Vert x_1 - y \Vert^2 + h_1 \leq \Vert x_2 - y\Vert^2 + h_2 \right\}\\
        H_2 &= \left\{y \in \RR^d: \Vert x_2 - y \Vert^2 + h_2 \leq \Vert x_1 - y\Vert^2 + h_1 \right\}.
    \end{align*}
    Additionally, let $\hat{H}_1 = H_1 \times (0,\infty)$ and $\hat{H}_2 = H_2 \times (0,\infty)$. The sets $\hat{H}_1$ and $\hat{H}_2$ have disjoint interiors, they only intersect at their boundaries. Because $\EE(\eta(\hat{H}_1 \cap \hat{H}_2)) = 0$, we have $\eta(\hat{H}_1 \cap \hat{H}_2) = 0$ almost surely. As a consequence, the Poisson processes $\eta_{\hat{H}_1}$ and $\eta_{\hat{H}_2}$ are independent. Here, $\eta_{\hat{H}_1}$ and $\eta_{\hat{H}_2}$ denote the restrictions of $\eta$ to $\hat{H}_1$ and $\hat{H}_2$ respectively. Note that $C((x_1,h_1),\eta+\delta_{(x_2,h_2)}) \subset C((x_1,h_1),\eta_{\hat{H}_1}+\delta_{(x_2,h_2)})$. Because $\eta_{\hat{H}_1}$ and $\eta_{\hat{H}_2}$ are independent, the random variables $\nu_d(C((x_1,h_1),\eta_{\hat{H}_1}+\delta_{(x_2,h_2)}))$ and $\nu_d(C((x_2,h_2),\eta_{\hat{H}_2}+\delta_{(x_1,h_1)}))$ are also independent. As a consequence:
    \begin{align}
        \EE&\left(\nu_d(C((x_1,h_1),\eta+\delta_{(x_2,h_2)}))\nu_d(C((x_2,h_2),\eta+\delta_{(x_1,h_1)})) \right) \leq \nonumber \\
        &\leq \EE\left(\nu_d(C((x_1,h_1),\eta_{\hat{H}_1}+\delta_{(x_2,h_2)}))\nu_d(C((x_2,h_2),\eta_{\hat{H}_2}+\delta_{(x_1,h_1)})) \right) \nonumber \\
        &= \EE\left(\nu_d(C((x_1,h_1),\eta_{\hat{H}_1}+\delta_{(x_2,h_2)}))\right)\EE\left(\nu_d(C((x_2,h_2),\eta_{\hat{H}_2}+\delta_{(x_1,h_1)})) \right). \label{expectation_product_volumes}
    \end{align}
    We will now derive a bound for the first expectation in (\ref{expectation_product_volumes}). By definition of the Laguerre cell, note that $C((x_1,h_1),\eta_{\hat{H}_1}+\delta_{(x_2,h_2)}) = C((x_1,h_1),\eta_{\hat{H}_1}) \cap H_1$. Let $\bar{x}_1$ be the orthogonal projection of $x_1$ onto $H_1$. Then, for all $y \in H_1$: $\Vert \bar{x}_1 - y\Vert \leq \Vert x_1 - y\Vert$. Therefore, we obtain the following inclusion:
    \begin{align}
        C\left((x_1,h_1),\eta_{\hat{H}_1}\right) \cap H_1 &= \left\{y \in H_1: \Vert x_1 - y\Vert^2 + h_1 \leq  \Vert x' - y\Vert^2 + h' \ \forall (x',h') \in \eta_{\hat{H}_1} \right\} \nonumber\\
        &\subseteq \left\{y \in H_1: \Vert \bar{x}_1 - y\Vert^2 + h_1 \leq  \Vert x' - y\Vert^2 + h' \ \forall (x',h') \in \eta_{\hat{H}_1} \right\}\nonumber \\
        &= C\left((\bar{x}_1,h_1),\eta_{\hat{H}_1}\right) \cap H_1. \nonumber
    \end{align}
    Hence,
    \begin{align}
        \EE\left(\nu_d\left(C\left((x_1,h_1),\eta_{\hat{H}_1}+\delta_{(x_2,h_2)}\right)\right)\right) &=  \EE\left(\nu_d\left(C\left((x_1,h_1),\eta_{\hat{H}_1}\right)\cap H_1\right)\right)\nonumber \\
        &\leq  \EE\left(\nu_d\left(C\left((\bar{x}_1,h_1),\eta_{\hat{H}_1}\right)\cap H_1\right)\right) \label{cell_volume_projection_bound}.
    \end{align}
    Write:
    \[B_{x,h,y} = \left\{(x',h')\in \RR^d \times (0,\infty): \Vert x - y \Vert^2 +h > \Vert x- x'\Vert^2 + h' \right\}\]
    Then, $y \in C((x,h),\eta) \iff \eta(B_{x,h,y}) = 0$. Via (\ref{Robbins_thm}) we now compute the expectation in (\ref{cell_volume_projection_bound}) as follows:
    \begin{align}
        \EE\left(\nu_d\left(C\left((\bar{x}_1,h_1),\eta_{\hat{H}_1}\right)\cap H_1\right)\right) &= \int_{\RR^d} \PP\left(y \in C\left((\bar{x}_1,h_1),\eta_{\hat{H}_1}\right)\cap H_1\right)\mathrm{d}y \nonumber \\
        &= \int_{H_1} \PP\left(y \in C\left((\bar{x}_1,h_1),\eta_{\hat{H}_1}\right)\right)\mathrm{d}y \nonumber \\
        &=\int_{H_1} \PP\left(\eta_{\hat{H}_1}\left(B_{\bar{x}_1,h_1,y} \right) = 0 \right)\mathrm{d}y \label{capped_cell_expected_volume}
    \end{align}
    Because $\eta_{\hat{H}_1}\left(B_{\bar{x}_1,h_1,y} \right)$ is Poisson distributed we can compute the probability of it being zero as:
    {\allowdisplaybreaks
    \begin{align}
        \PP&\left(\eta_{\hat{H}_1}\left(B_{\bar{x}_1,h_1,y} \right) = 0 \right) = \nonumber\\
        &=\exp\left(-\int_{H_1}\int_0^\infty \mathds{1}\left\{\Vert \bar{x}_1 - x'\Vert^2 + t < \Vert\bar{x}_1 - y \Vert^2 + h_1\right\}\mathrm{d}F(t)\mathrm{d}x'\right) \nonumber \\
        &= \exp\left(-\int_{H_1}\int_0^{\Vert\bar{x}_1 - y \Vert^2 + h_1} \mathds{1}\left\{\Vert \bar{x}_1 - x'\Vert < \sqrt{\Vert\bar{x}_1 - y \Vert^2 + h_1 - t}\right\}\mathrm{d}F(t)\mathrm{d}x'\right) \nonumber\\
        &= \exp\left(-\int_0^{\Vert\bar{x}_1 - y \Vert^2 + h_1} \nu_d\left(\bar{B}\left(\bar{x}_1,\sqrt{\Vert\bar{x}_1 - y \Vert^2 + h_1 - t} \right)\cap H_1 \right) \mathrm{d}F(t)\right) \label{capped_cell_inclusion_probability}
    \end{align}
    }%
    Because $\bar{x}_1 \in H_1$, if we consider the volume of the intersection of the half space $H_1$ and the ball centered at $\bar{x}_1$ as in (\ref{capped_cell_inclusion_probability}), it follows that the volume of that intersection is at least half the volume of said ball. Hence:
    \[\PP\left(\eta_{\hat{H}_1}\left(B_{\bar{x}_1,h_1,y} \right) = 0 \right) \leq \exp\left(-\frac{\kappa_d}{2}\int_0^{\Vert\bar{x}_1 - y \Vert^2 + h_1}\left(\Vert \bar{x}_1  - y\Vert^2 + h_1 - t \right)^\frac{d}{2}  \mathrm{d}F(t)\right).\]
    Plugging this bound back into (\ref{capped_cell_expected_volume}) we obtain:
    \begin{align}
        \EE\left(\nu_d\left(C\left((\bar{x}_1,h_1),\eta_{\hat{H}_1}\right)\cap H_1\right)\right) &\leq \int_{\RR^d} \exp\left(-\frac{\kappa_d}{2}\int_0^{\Vert\bar{x}_1 - y \Vert^2 + h_1}\left(\Vert \bar{x}_1  - y\Vert^2 + h_1 - t \right)^\frac{d}{2}  \mathrm{d}F(t)\right)\mathrm{d}y \nonumber \\
        &= \int_{\RR^d} \exp\left(-\frac{\kappa_d}{2}\int_0^{\Vert u \Vert^2 + h_1}\left(\Vert u\Vert^2 + h_1 - t \right)^\frac{d}{2}  \mathrm{d}F(t)\right)\mathrm{d}u \label{capped_cell_final_bound}
    \end{align}
    The final equality follows from substituting $u = y - \bar{x}_1$. Note that the bound in (\ref{capped_cell_final_bound}) does not depend on $x_1$, $x_2$ and $h_2$. By symmetry, replacing $h_1$ with $h_2$ in (\ref{capped_cell_final_bound}) serves as an upper bound for the second expectation in (\ref{expectation_product_volumes}). We now conclude:
    \begin{align*}
        \int_0^\infty &\int_0^\infty \EE\left(\nu_d(C((x_1,h_1),\eta+\delta_{(x_2,h_2)}))\nu_d(C((x_2,h_2),\eta+\delta_{(x_1,h_1)})) \right)\mathrm{d}F(h_1)\mathrm{d}F(h_2)  \\
        &\leq \left(\int_0^\infty \int_{\RR^d} \exp\left(-\frac{\kappa_d}{2}\int_0^{\Vert u \Vert^2 + h}\left(\Vert u\Vert^2 + h - t \right)^\frac{d}{2}  \mathrm{d}F(t)\right)\mathrm{d}u\mathrm{d}F(h)\right)^2 \\
        &= 2^2.
    \end{align*}
    The fact that the final integral equals 2 follows from (\ref{mean_cell_volume1}) by substituting $\tilde{F} = F/2$.
\end{proof}

\begin{proof}[Proof of Theorem \ref{thm_an_L2_bound}]
    For $j \in \{1,\dots,d\}$, let $A_{n,j}$ denote the $j$-th component of the random vector $A_n$. Because $\EE(A_n) = 0$ we may write:
    \begin{equation}
        \EE\left( \left\Vert A_n \right\Vert^2 \right) = \EE\left(A_{n,1}^2 + A_{n,2}^2 + \dots + A_{n,d}^2\right) = \sum_{j=1}^d \mathrm{Var}(A_{n,j}).\label{sum_squared_an}
    \end{equation}
    We will apply Lemma \ref{poincare_inequality_lemma} to $\mathrm{Var}(A_{n,j})$ for each $j$. The fact that $\EE(A_{n,j}^2) < \infty$ follows from (\ref{an_simple_l2_bound}). Fix $j \in \{1,\dots,d\}$ and set $f(\eta)=A_{n,j}$. Let $(x',h') \in \RR^d \times(0,M)$. For $(x,h)\in \eta$ let $R_x(\eta)$ be as in Theorem \ref{stabilization_theorem}, via property 1 in Theorem \ref{stabilization_theorem} we may write:
    \begin{align*}
        f(\eta+\delta_{(x',h')}) &=\frac{1}{\nu_d(W_n)}\sum_{(x,h) \in \eta+\delta_{(x',h')}}\mathds{1}_{W_n}(x)\nu_d(C((x,h),\eta+\delta_{(x',h')}))x \\
        &= \frac{1}{\nu_d(W_n)}\mathds{1}_{W_n}(x')\nu_d(C((x',h'),\eta))x' + \\
        &\phantom{=} +\frac{1}{\nu_d(W_n)}\sum_{(x,h) \in \eta}\mathds{1}_{W_n}(x)\nu_d(C((x,h),\eta+\delta_{(x',h')}))x \\
        &= \frac{1}{\nu_d(W_n)}\mathds{1}_{W_n}(x')\nu_d(C((x',h'),\eta))x' + \\
        &\phantom{=} +\frac{1}{\nu_d(W_n)}\sum_{(x,h) \in \eta}\mathds{1}_{W_n}(x)\mathds{1}\{\Vert x - x'\Vert \leq R_x(\eta)\}\nu_d((C(x,h),\eta+\delta_{(x',h')}))x \\
        &\phantom{=} +\frac{1}{\nu_d(W_n)}\sum_{(x,h) \in \eta}\mathds{1}_{W_n}(x)\mathds{1}\{\Vert x - x'\Vert > R_x(\eta)\}\nu_d(C((x,h),\eta))x.
    \end{align*}
    Similarly, we may write:
    \begin{align*}
        f(\eta) &= \frac{1}{\nu_d(W_n)}\sum_{(x,h) \in \eta}\mathds{1}_{W_n}(x)\mathds{1}\{\Vert x - x'\Vert > R_x(\eta)\}\nu_d(C((x,h),\eta))x + \\
        &\phantom{=} + \frac{1}{\nu_d(W_n)}\sum_{(x,h) \in \eta}\mathds{1}_{W_n}(x)\mathds{1}\{\Vert x - x'\Vert \leq R_x(\eta)\}\nu_d(C((x,h),\eta))x.
    \end{align*}
    Using these expressions for $f(\eta+\delta_{(x',h')})$ and $f(\eta)$ we derive an upper bound for\\ $|D_{(x',h')}f(\eta)|$. Taking into account that $\Vert x\Vert \leq \mathrm{diam}(W)n^{1/d}$, we obtain:
    {\allowdisplaybreaks
    \begin{align*}
        |D_{(x',h')}f(\eta)|  &\leq \frac{1}{n}\mathds{1}_{W_n}(x')\nu_d(C((x',h'),\eta))\Vert x'\Vert + \\
        &\phantom{\leq} + \frac{1}{n}\sum_{(x,h) \in \eta}\mathds{1}_{W_n\cap \bar{B}(x',R_x(\eta))}(x)\left|\nu_d(C((x,h),\eta)) -\nu_d(C((x,h),\eta+\delta_{(x',h')})) \right|\Vert x \Vert \\
        &\leq \frac{\mathrm{diam}(W)}{n^{1-\frac{1}{d}}}\mathds{1}_{W_n}(x')\nu_d(C(x',h'),\eta) +\\
        &\phantom{\leq} + \frac{\mathrm{diam}(W)}{n^{1-\frac{1}{d}}}\sum_{(x,h) \in \eta}\mathds{1}_{W_n}(x)\mathds{1}\{\Vert x - x'\Vert \leq R_x(\eta)\}\nu_d(C((x,h),\eta)).
    \end{align*}}%
    Here, we also used the fact: $C((x,h),\eta+\delta_{(x',h')}) \subset C((x,h),\eta)$. Note that this upper bound does not depend on $j$. Via Lemma \ref{poincare_inequality_lemma} and (\ref{sum_squared_an}) we obtain:
    \begin{align}
    \begin{split}
        \EE\left( \left\Vert A_n \right\Vert^2 \right) &\leq \frac{\mathrm{diam}(W)^2d}{n^{2-\frac{2}{d}}}\Bigg(\int_{\RR^d}\int_0^M \mathds{1}_{W_n}(x')\EE\left(\nu_d(C((x',h'),\eta))^2 \right)\mathrm{d}F(h')\mathrm{d}x' \\
        &\phantom{\leq} \ + \int_{\RR^d}\int_0^M \EE\Bigg(\bigg(\sum_{(x,h) \in \eta}\mathds{1}_{W_n}(x)\mathds{1}\{\Vert x - x'\Vert \leq R_x(\eta)\}\nu_d(C((x,h),\eta)) \bigg)^2 \Bigg)\mathrm{d}F(h')\mathrm{d}x'\Bigg). \label{an_L2_bound}
    \end{split}
    \end{align}
    We now separately determine upper bounds for both of the integrals appearing in (\ref{an_L2_bound}). For the first integral we obtain the following:
    \begin{equation}
        \int_{\RR^d}\int_0^M \mathds{1}_{W_n}(x')\EE\left(\nu_d(C((x',h'),\eta))^2 \right)\mathrm{d}F(h')\mathrm{d}x' \leq \int_{\RR^d} \mathds{1}_{W_n}(x')\beta_2\mathrm{d}x' = \beta_2n.\label{an_proof_bound1}
    \end{equation}
    Here, $\beta_2$ is as in Theorem \ref{thm_integrated_second_moment}. Let $(\eta)_{\neq}^2$ denotes the set of all distinct pairs of points of $\eta$. That is, if $(x_1,h_1),(x_2,h_2) \in (\eta)_{\neq}^2$, then $(x_1,h_1) \neq (x_2,h_2)$. Now, consider the second integral in (\ref{an_L2_bound}). Expanding the square and applying Fubini, we obtain:
    {\allowdisplaybreaks
    \begin{align}
    &\int_{\RR^d}\int_0^M \EE\Bigg(\bigg(\sum_{(x,h) \in \eta}\mathds{1}_{W_n}(x)\mathds{1}\{\Vert x - x'\Vert \leq R_x(\eta)\}\nu_d(C((x,h),\eta)) \bigg)^2 \Bigg)\mathrm{d}F(h')\mathrm{d}x' \nonumber\\
        &=\int_{\RR^d}\int_0^M \EE\left(\sum_{(x,h) \in \eta}\mathds{1}_{W_n}(x)\mathds{1}\{\Vert x - x'\Vert \leq R_x(\eta)\}\nu_d(C((x,h),\eta))^2 \right)\mathrm{d}F(h')\mathrm{d}x' +\nonumber\\
        &\phantom{=} + \int_{\RR^d}\int_0^M \EE\Bigg(\sum_{(x_1,h_1),(x_2,h_2) \in (\eta)_{\neq}^2}\mathds{1}_{W_n}(x_1)\mathds{1}_{W_n}(x_2)\mathds{1}\{\Vert x_1 - x'\Vert \leq R_{x_1}(\eta)\}\cdot \nonumber\\
        &\phantom{=} \cdot\mathds{1}\{\Vert x_2 - x'\Vert \leq R_{x_2}(\eta)\}\nu_d(C((x_1,h_1),\eta))\nu_d(C((x_2,h_2),\eta)) \Bigg)\mathrm{d}F(h')\mathrm{d}x'\nonumber\\
        &= F(M) \EE\left(\sum_{(x,h) \in \eta}\mathds{1}_{W_n}(x)\int_{\RR^d}\mathds{1}\{\Vert x - x'\Vert \leq R_x(\eta)\}\mathrm{d}x'\nu_d(C((x,h),\eta))^2 \right) \nonumber\\
        &\phantom{=} + F(M) \EE\Bigg(\sum_{(x_1,h_1),(x_2,h_2) \in (\eta)_{\neq}^2}\mathds{1}_{W_n}(x_1)\mathds{1}_{W_n}(x_2)\int_{\RR^d}\mathds{1}\{\Vert x_1 - x'\Vert \leq R_{x_1}(\eta)\}\cdot\nonumber\\
        &\phantom{=} \cdot\mathds{1}\{\Vert x_2 - x'\Vert \leq R_{x_2}(\eta)\}\mathrm{d}x'\nu_d(C((x_1,h_1),\eta))\nu_d(C((x_2,h_2),\eta)) \Bigg)\nonumber\\
        &= F(M)\kappa_d \EE\left(\sum_{(x,h) \in \eta}\mathds{1}_{W_n}(x)R_x(\eta)^d\nu_d(C((x,h),\eta))^2 \right) \label{an_proof_term_squares}\\
        \begin{split}
            &\phantom{=} + F(M) \EE\Bigg(\sum_{(x_1,h_1),(x_2,h_2) \in (\eta)_{\neq}^2}\mathds{1}_{W_n}(x_1)\mathds{1}_{W_n}(x_2)\nu_d(\bar{B}(x_1,R_{x_1}(\eta))\cap\bar{B}(x_2,R_{x_2}(\eta)))\cdot \label{an_proof_term_cross} \\
        &\phantom{=}\cdot\nu_d(C((x_1,h_1),\eta))\nu_d(C((x_2,h_2),\eta)) \Bigg).
        \end{split}
    \end{align}
    }%
    An upper bound for the term in (\ref{an_proof_term_squares}) may be computed via the Mecke equation as follows:
    \begin{align}
        F(M)\kappa_d &\EE\left(\sum_{(x,h) \in \eta}\mathds{1}_{W_n}(x)R_x(\eta)^d\nu_d(\bar{B}(x,R_x(\eta)))^2 \right) \leq \nonumber \\
        &\leq F(M)\kappa_d^3 \EE\left(\sum_{(x,h) \in \eta}\mathds{1}_{W_n}(x)R_x(\eta)^{3d}\right) \nonumber\\
        &= F(M)\kappa_d^3\int_{\RR^d}\int_0^M\mathds{1}_{W_n}(x)\EE(R_x(\eta+\delta_x)^{3d}) \mathrm{d}F(h)\mathrm{d}x \nonumber\\
        &\leq F(M)^2\kappa_d^3\EE\left(R_0(\eta)^{3d}\right)n. \label{an_proof_bound2}
    \end{align}
    Here, we also used point 4 of Theorem \ref{stabilization_theorem}. Because $R_0(\eta)$ has exponentially decaying tails by point 3 of Theorem \ref{stabilization_theorem}, all of its moments are finite, hence $\EE(R_0(\eta)^{3d}) < \infty$. Next, we consider the term in (\ref{an_proof_term_cross}). This term may be computed via the multivariate Mecke equation (Theorem 4.4 in \cite{Last2018}). By also using points 2 and 4 of Theorem \ref{stabilization_theorem} we obtain:
{\allowdisplaybreaks
    \begin{align}
        &\phantom{=}F(M) \EE\Bigg(\sum_{(x_1,h_1),(x_2,h_2) \in (\eta)_{\neq}^2}\mathds{1}_{W_n}(x_1)\mathds{1}_{W_n}(x_2)\nu_d(\bar{B}(x_1,R_{x_1}(\eta))\cap\bar{B}(x_2,R_{x_2}(\eta)))\cdot \nonumber\\
        &\phantom{=}\cdot\nu_d(C((x_1,h_1),\eta))\nu_d(C((x_2,h_2),\eta)) \Bigg)\nonumber\\
        &\leq F(M)\kappa_d^2\EE\Bigg(\sum_{(x_1,h_1),(x_2,h_2) \in (\eta)_{\neq}^2}\mathds{1}_{W_n}(x_1)\mathds{1}_{W_n}(x_2)\nu_d(\bar{B}(x_1,R_{x_1}(\eta))\cap\bar{B}(x_2,R_{x_2}(\eta)))\cdot \nonumber\\
        &\phantom{=}\cdot R_{x_1}(\eta)^d  R_{x_2}(\eta)^d \Bigg)\nonumber\\
        &= F(M)\kappa_d^2 \int_{W_n}\int_0^M \int_{W_n}\int_0^M \EE\Bigg(\nu_d(\bar{B}(x_1,R_{x_1}(\eta + \delta_{x_2}))\cap\bar{B}(x_2,R_{x_2}(\eta + \delta_{x_1})))\cdot\nonumber \\
        &\phantom{=}\cdot R_{x_1}(\eta + \delta_{x_2})^d  R_{x_2}(\eta+ \delta_{x_1})^d \Bigg)\mathrm{d}F(h_1)\mathrm{d}x_1\mathrm{d}F(h_2)\mathrm{d}x_2\nonumber\\
        &\leq F(M)^3\kappa_d^2 \int_{W_n}\int_{W_n} \EE\Bigg(\nu_d(\bar{B}(x_1,R_{x_1}(\eta))\cap\bar{B}(x_2,R_{x_2}(\eta))) R_{x_1}(\eta)^d  R_{x_2}(\eta)^d \Bigg)\mathrm{d}x_1\mathrm{d}x_2 \label{an_proof_cross_part2} 
    \end{align}}%
    Note that for any $x_1,x_2 \in \RR^d$:
    \begin{align*}
        \nu_d(\bar{B}(x_1,R_{x_1}(\eta))\cap\bar{B}(x_2,R_{x_2}(\eta))) &= \nu_d(\bar{B}(x_1,R_{x_1}(\eta))\cap\bar{B}(x_2,R_{x_2}(\eta))) \cdot \\
        &\phantom{=} \ \cdot\mathds{1}\{\Vert x_1 - x_2 \Vert \leq R_{x_1}(\eta) + R_{x_2}(\eta)\}
    \end{align*}
    As a consequence, (\ref{an_proof_cross_part2}) may be written as:
    {\allowdisplaybreaks
    \begin{align}
        &F(M)^3\kappa_d^2 \int_{W_n}\int_{W_n} \EE\Bigg(\nu_d(\bar{B}(x_1,R_{x_1}(\eta))\cap\bar{B}(x_2,R_{x_2}(\eta))) R_{x_1}(\eta)^d  R_{x_2}(\eta)^d\cdot \nonumber \\
        &\phantom{=} \cdot \mathds{1}\{\Vert x_1 - x_2 \Vert \leq R_{x_1}(\eta) + R_{x_2}(\eta)\} \Bigg)\mathrm{d}x_1\mathrm{d}x_2 \nonumber \\
        \begin{split}\label{an_term_1}
            &= F(M)^3\kappa_d^2 \int_{W_n}\int_{W_n} \EE\Bigg(\nu_d(\bar{B}(x_1,R_{x_1}(\eta))\cap\bar{B}(x_2,R_{x_2}(\eta))) R_{x_1}(\eta)^d  R_{x_2}(\eta)^d\cdot \\
        &\phantom{=} \cdot \mathds{1}\{\Vert x_1 - x_2 \Vert \leq R_{x_1}(\eta) + R_{x_2}(\eta)\} \mathds{1}\{R_{x_1}(\eta) \leq R_{x_2}(\eta)\}\Bigg)\mathrm{d}x_1\mathrm{d}x_2 +
        \end{split}\\
        \begin{split}\label{an_term_2}
            &\phantom{=} + F(M)^3\kappa_d^2 \int_{W_n}\int_{W_n} \EE\Bigg(\nu_d(\bar{B}(x_1,R_{x_1}(\eta))\cap\bar{B}(x_2,R_{x_2}(\eta))) R_{x_1}(\eta)^d  R_{x_2}(\eta)^d\cdot \\
        &\phantom{=} \cdot \mathds{1}\{\Vert x_1 - x_2 \Vert \leq R_{x_1}(\eta) + R_{x_2}(\eta)\} \mathds{1}\{R_{x_1}(\eta) > R_{x_2}(\eta)\}\Bigg)\mathrm{d}x_1\mathrm{d}x_2
        \end{split}     
    \end{align}
    }%
We now derive an upper bound for the expression in (\ref{an_term_1}), the obtained upper bound will also hold for (\ref{an_term_2}) by symmetry in $x_1$ and $x_2$. The following is an upper bound for (\ref{an_term_1}):
\begin{align}
    &\phantom{\leq} F(M)^3\kappa_d^2 \int_{W_n}\int_{W_n} \EE\Bigg(\nu_d(\bar{B}(x_2,R_{x_2}(\eta)))  R_{x_2}(\eta)^{2d}\mathds{1}\{\Vert x_1 - x_2 \Vert \leq 2R_{x_2}(\eta)\}\cdot \nonumber \\
        &\phantom{=} \cdot  \mathds{1}\{R_{x_1}(\eta) \leq R_{x_2}(\eta)\}\Bigg)\mathrm{d}x_1\mathrm{d}x_2\nonumber \\
    &\leq F(M)^3\kappa_d^3 \int_{W_n} \EE\Bigg(R_{x_2}(\eta)^{3d}\int_{\RR^d}\mathds{1}\{\Vert x_1 - x_2 \Vert \leq 2R_{x_2}(\eta)\}\mathrm{d}x_1\Bigg)\mathrm{d}x_2\nonumber\\
    &= F(M)^3\kappa_d^4 2^d \EE(R_0(\eta)^{4d}) n.\label{an_proof_bound3}
\end{align}
Plugging (\ref{an_proof_bound1}), (\ref{an_proof_bound2}) and (\ref{an_proof_bound3}) back into (\ref{an_L2_bound}) yields the desired result.
\[ \EE\left( \left\Vert A_n \right\Vert^2 \right) \leq \frac{\mathrm{diam}(W)^2d}{n^{1-\frac{2}{d}}}\left(\beta_2 + F(M)^2\kappa_d^3\EE\left(R_0(\eta)^{3d}\right) + F(M)^3\kappa_d^4 2^{d+1} \EE\left(R_0(\eta)^{4d}\right) \right).\]
\end{proof}

\bibliographystyle{abbrv}
\addcontentsline{toc}{section}{References}
\bibliography{export}

\newpage

\end{document}